\pdfoutput=1
\documentclass[11pt]{article}

\usepackage[letterpaper,margin=1in]{geometry}

\usepackage{lipsum}

\newcommand\blfootnotea[1]{%
  \begingroup
  \renewcommand\thefootnote{}\footnote{#1}%
  \endgroup
}

\usepackage[utf8]{inputenc} %
\usepackage[T1]{fontenc}    %

\usepackage{url}            %
\usepackage{booktabs}       %
\usepackage{nicefrac}       %
\usepackage{microtype}      %
\usepackage{enumitem}
\usepackage{dsfont}
\usepackage{multicol}
\usepackage{xcolor}

\usepackage{amsthm,amsfonts,amsmath,amssymb,epsfig,color,float,graphicx,verbatim, enumitem}

\usepackage{algpseudocode,algorithm,algorithmicx}  

\usepackage{mathtools}
\usepackage{bbm}

\usepackage{thm-restate}
\usepackage{turnstile}

\usepackage[colorlinks,citecolor=blue,linkcolor=magenta,bookmarks=true,hypertexnames=false]{hyperref}
\usepackage[nameinlink]{cleveref}

\crefname{equation}{equation}{equations}
\crefname{lemma}{lemma}{lemmata}
\crefname{claim}{claim}{claims}
\crefname{theorem}{theorem}{theorems}
\crefname{proposition}{proposition}{propositions}
\crefname{corollary}{corollary}{corollaries}
\crefname{claim}{claim}{claims}
\crefname{remark}{remark}{remarks}
\crefname{definition}{definition}{definitions}
\crefname{fact}{fact}{facts}
\crefname{question}{question}{questions}
\crefname{condition}{condition}{conditions}
\crefname{algorithm}{algorithm}{algorithms}
\crefname{assumption}{assumption}{assumptions}
\crefname{problem}{problem}{problems}

\newtheorem{theorem}{Theorem}[section]
\newtheorem{lemma}[theorem]{Lemma}

\newtheorem{definition}[theorem]{Definition}

\newtheorem{fact}[theorem]{Fact}

\theoremstyle{definition}

\newtheorem{remark}[theorem]{Remark}
\newtheorem{problem}[theorem]{Problem}

\newcommand{\poly}{\mathrm{poly}}
\newcommand{\polylog}{\mathrm{polylog}}

\newcommand{\Var}{\mathbf{Var}}

\def\R{\mathbb R}

\def\N{\mathbb N}

\def\Z{\mathbb Z}

\newcommand{\cA}{\mathcal{A}}

\newcommand{\cD}{\mathcal{D}}

\newcommand{\cL}{\mathcal{L}}

\newcommand{\cN}{\mathcal{N}}

\newcommand{\cP}{\mathcal{P}}

\newcommand{\cU}{\mathcal{U}}

\newcommand*{\pE}{\tilde{\E}}

\newcommand{\Paren}[1]{\left(#1\right)}

\newcommand{\Brac}[1]{\left[#1\right]}

\newcommand{\Set}[1]{\left\{#1\right\}}

\newcommand{\norm}[1]{\lVert#1\rVert}
\newcommand{\Norm}[1]{\left\lVert#1\right\rVert}

\newcommand{\iprod}[1]{\langle#1\rangle}

\newcommand{\hide}[1]{}

\DeclareMathOperator*{\pr}{\mathrm{Pr}}
\DeclareMathOperator*{\E}{\mathbf{E}}

\def\d{\mathrm{d}}

\newcommand*{\ak}{\cA_{\text{$k$-sparse}}}

\def\colorful{1}

\ifnum\colorful=1

\else

\fi

\allowdisplaybreaks

\title{List-Decodable Sparse Mean Estimation \\ via Difference-of-Pairs Filtering\blfootnotea{Author names are listed in alphabetical order.}}

\author{
Ilias Diakonikolas\thanks{Supported by NSF Medium Award CCF-2107079,
NSF Award CCF-1652862 (CAREER), a Sloan Research Fellowship, and
a DARPA Learning with Less Labels (LwLL) grant.}\\
University of Wisconsin-Madison\\
{\tt ilias@cs.wisc.edu}\\
\and
Daniel M. Kane\thanks{Supported by NSF Medium Award CCF-2107547,
NSF Award CCF-1553288 (CAREER), a Sloan Research Fellowship, and a grant from CasperLabs.}\\
University of California, San Diego\\
{\tt dakane@cs.ucsd.edu}
\and
Sushrut Karmalkar\thanks{Supported by NSF under Grant \#2127309 to the Computing Research Association for the CIFellows 2021 Project.}\\
University of Wisconsin-Madison\\
{\tt skarmalkar@wisc.edu}\\
\and
Ankit Pensia\thanks{Supported by NSF grants NSF Award CCF-1652862 (CAREER), DMS-1749857, 
and CCF-1841190.}\\
University of Wisconsin-Madison\\
{\tt ankitp@cs.wisc.edu}\\
\and
Thanasis Pittas\thanks{Supported by NSF Award CCF-1652862 (CAREER) and 
NSF Medium Award CCF-2107079.}\\
University of Wisconsin-Madison\\
{\tt pittas@wisc.edu}\\
}
 
\begin{document}

\maketitle

\begin{abstract}%
We study the problem of list-decodable \emph{sparse} mean estimation. 
Specifically, for a parameter $\alpha \in (0, 1/2)$, we are given $m$ points in $\R^n$, 
$\lfloor \alpha m \rfloor$ of which are i.i.d.\ samples from a distribution $D$ 
with unknown $k$-sparse mean $\mu$. 
No assumptions are made on the remaining points, 
which form the majority of the dataset. The goal is to return a small list of candidates 
containing a vector $\hat \mu$ such that $\norm{\hat \mu - \mu}_2$ is small. 
Prior work had studied the problem of list-decodable mean estimation in the dense setting. 
In this work, we develop a novel, conceptually simpler technique 
for list-decodable mean estimation. As the main application of our approach, 
we provide the first sample and computationally efficient algorithm 
for list-decodable sparse mean estimation. In particular, for distributions with  
``certifiably bounded'' $t$-th moments in $k$-sparse directions 
and sufficiently light tails, our algorithm achieves error of $(1/\alpha)^{O(1/t)}$ 
with sample complexity $m = (k\log(n))^{O(t)}/\alpha$ and running time $\poly(mn^t)$. 
For the special case of Gaussian inliers, our algorithm achieves the optimal error guarantee of 
$\Theta (\sqrt{\log(1/\alpha)})$ with quasi-polynomial sample and computational complexity. 
We complement our upper bounds with nearly-matching statistical query 
and low-degree polynomial testing lower bounds. 
\end{abstract}

\setcounter{page}{0}

\thispagestyle{empty}

\newpage

\section{Introduction} %
\label{sec:introduction}

It is well-established that when a dataset is corrupted by outliers, 
many commonly-used estimators fail to produce reliable estimates~\cite{Tukey60,AndBHHRT72}.
The field of robust statistics was developed to perform reliable statistical inference 
in the presence of a constant fraction of outliers, 
even when the data is high-dimensional~\cite{Huber09,HampelEtalBook86}.
Although statistical rates of high-dimensional robust estimation problems 
are relatively well-understood by now~\cite{DonLiu88a,Yatracos85,Donoho92,Huber09,CheGR16}, 
all of the estimators developed in these works were computationally inefficient, 
with runtime exponential in the dimension. 
The goal of algorithmic robust statistics, beginning with the works of \cite{DKKLMS16,LaiRV16}, 
is to design computationally efficient algorithms for high-dimensional robust estimation tasks.
We refer the reader to the survey~\cite{DK19-survey} for an introduction to this field.

The bulk of the recent progress in algorithmic robust statistics has focused 
on the setting where the fraction of outliers is a small constant, 
and the majority of samples are inliers, 
see, e.g.,~\cite{DKKLMS16, LaiRV16,KlivansKM18}. 
In contrast, when the fraction of outliers outnumbers the fraction of inliers, 
it is generally information-theoretically impossible to output 
a single estimate with non-vacuous error guarantees.  
In such situations, we allow the algorithm to return a small list of candidates 
such that one of the candidates is close to the true parameter. 
This \emph{list-decodable} learning setting was first introduced 
in \cite{BBV08} and developed in~\cite{CSV17}. We define the model below.

\begin{definition}[List-Decodable Learning]\label{def:LDL}
Given a parameter $0 < \alpha < 1/2$ and a distribution family $\cD$ on $\R^n$, 
the algorithm specifies $m \in \Z_+$ and observes a set of $m$ samples constructed as follows: 
First, a set $S$ of $\lfloor \alpha m \rfloor$ i.i.d.\ samples are drawn from an (unknown) 
distribution $D \in \cD$. Then, an adversary is allowed to inspect $S$ 
and choose a multiset $E$  of $m - \lfloor \alpha m \rfloor$ points. 
The multiset $T$, defined as  $T:= S \cup E$, of $m$ points is given as input to the algorithm. 
We say that $D$ is the distribution of inliers, the elements in $S$ are inliers, 
the points in $E$ are outliers, and $T$ is an $(1-\alpha)$-corrupted dataset of $S$. 
The goal is to output a ``small'' list of hypotheses $\cL$ 
at least one of which is (with high probability) close to the target  parameter of $D$. 
\end{definition}

The list-decodable learning setting, interesting in its own right, 
is closely related to several well-studied problems. 
A natural example is the problem of parameter recovery 
from mixture models, for instance, Gaussian mixtures 
(see, e.g.,~\cite{Dasgupta:99,vempala2004spectral,arora2005learning,
dasgupta2007probabilistic,kumar2010clustering,regev2017learning}).
List-decodable mean estimation can serve as a key step in learning mixtures, 
since one can treat any component of the mixture as the set of inliers 
(see, e.g.,~\cite{CSV17,DKS18-list,KStein17}).
In addition, list-decodable learning can be used to model data 
in important applications where mixture models are not sufficient, 
such as crowdsourcing (see, e.g.,~\cite{svc16,SteinhardtKL17,MeisterV18}) 
and community detection in stochastic block models (e.g.,~\cite{CSV17}). 

Prior work on list-decodable mean estimation has focused on the unstructured setting, 
where the target mean is an arbitrary dense vector 
(see, e.g.,~\cite{CSV17,KStein17,DKS18-list,raghavendra2020list, cherapanamjeri2020list, diakonikolas2020list, diakonikolas2021list, diakonikolas2022list}). 
Sparse models have proven to be useful in a wide range of statistical tasks, 
and thus understanding the statistical and computational 
aspects of sparse estimation is a fundamental problem 
(see, e.g.,~\cite{eldar2012compressed,Hastie15,vandeGeer16}).
Here we study list-decodable {\em sparse} mean estimation, 
where the target mean vector is known to be {\em $k$-sparse}, 
i.e., it has at most $k$ non-zero coordinates. Given an $(1-\alpha)$-corrupted set of samples, 
our goal is to output, in a computationally-efficient manner, 
a small list of vectors containing a good approximation $\widehat{\mu}$ 
to the true mean $\mu$ (cf.\ \Cref{def:LDL}). 
{Importantly, the goal is to achieve this with \emph{far fewer samples} than in the dense setting. 
While the dense setting would require sample size polynomial in $n$ --- the ambient dimension of the data --- 
the goal here is to solve the problem in number of samples polynomial in $k$ and only polylogarithmic in $n$. }

In this paper, we present a novel and conceptually simple technique for 
list-decodable mean estimation (that is applicable even in the dense setting). 
Combining our framework with the concentration results from  \cite{DKKPP22}, 
we obtain the first sample and computationally efficient algorithm for list-decodable sparse mean estimation. 
We note that, while prior results~\cite{KStein17,raghavendra2020list} 
can possibly be modified to incorporate the sparsity framework of \cite{DKKPP22} 
after sufficient effort, a notable contribution of our work is a general and 
conceptually simpler framework for list-decodable estimation, 
which can easily be adapted to incorporate various structural constraints.

\subsection{Related Work}
Efficient estimators for high-dimensional robust statistics 
with sparsity constraints have been recently developed for various problems, 
such as mean estimation and PCA (see, e.g.,~\cite{BDLS17, DKKPS19-sparse, Li17-sparse, DKKPP22}).
The problem of list-decodable mean estimation was first introduced in~\cite{CSV17}, 
in which the authors achieved an error guarantee of $\tilde{O}(1/\sqrt{\alpha})$ 
for distributions with bounded second moment; 
this guarantee turned out to be optimal for this distributional assumption (see~\cite{DKS18-list}). 
Subsequent work improved the algorithmic guarantees 
for this problem (see, e.g.,~\cite{cherapanamjeri2020list, diakonikolas2020list, diakonikolas2021list}).

To achieve better error guarantees, it is necessary to make further assumptions 
on the distribution, e.g., Gaussianity or bounded higher moments. 
In terms of the minimax optimal rate\footnote{Informally speaking, 
we say that the minimax optimal rate is $\gamma$ if (i) no algorithm 
(regardless of sample size and runtime) has error $o(\gamma)$ 
with a list of size independent of dimension $n$, 
and (ii) there is an algorithm with error $O(\gamma)$ with a  list of size independent of $n$; 
in our case, $\poly(n/\alpha)$ samples and $O(1/\alpha)$ list size suffice.}, 
\cite{DKS18-list} showed that the optimal error for Gaussians is $\Theta(\sqrt{\log(1/\alpha)})$.
They also showed that any SQ algorithm that achieves the optimal error for Gaussians 
must take either super-polynomial time or samples, and presented an algorithm with matching guarantees.
When the distribution $D$ has bounded $t$-th moment for some even $t> 2$ (
i.e., $\E[\langle v, X - \E[X]\rangle^t]$ is bounded for all unit vectors $v$), 
\cite{DKS18-list} proved that the minimax optimal rate is $\Theta(\alpha^{-1/t})$.
For distributions with certifiably bounded $t$-th moments, 
\cite{KStein17, raghavendra2020list} provided algorithms obtaining 
an error rate of $O((1/\alpha)^{O(1/t)})$
with sample complexities $m = \poly(n^t/\alpha)$, 
and runtimes $\poly(m^tn^t)$ and $\poly(mn)^{\poly(t,1/\alpha)}$, respectively.
Recently, in the context of moment estimation and clustering problems, 
\cite{steurer2021sos} showed how to improve the dependence 
on $m$ in the runtime of algorithms 
that depend on certifiably bounded moments from $\poly(m^tn^t)$ to $\poly(mn^t)$. 
While it is possible that their result might be applied to \cite{KStein17,raghavendra2020list}, 
our algorithmic technique naturally lends itself to achieve runtime $\poly(mn^t)$ 
for the dense as well as sparse settings. 
We provide detailed comparisons with (see, e.g.,~\cite{DKS18-list,KStein17,raghavendra2020list}) in \Cref{sec:our-tech}.
Finally, we mention that the list-decodable setting has also been studied 
in the context of  linear regression (see, e.g.,~\cite{karmalkar2019list,raghavendra2020list,diakonikolas2021statistical}) 
and subspace recovery (see, e.g.,~\cite{BK20-subspace,RY20-subspace}). %

\subsection{Our Results} \label{ssec:results}

{We demonstrate an algorithm to perform list-decodable sparse 
	mean estimation with $(k \log n)^{O(t)}$ samples, 
	when the mean $\mu \in \R^n$ is known to be $k$-sparse.} 
For this to be possible, we will require some assumptions 
on the underlying distribution of inliers $D$. 
Prior work in the dense setting (\cite{KStein17,raghavendra2020list}) 
assumed that the  inlier distribution $D$ in~\Cref{def:LDL} 
satisfies \emph{$d$-certifiably} bounded $t$-th moments in every direction 
(i.e., for some moment bound $M > 0$, 
$M \|v\|_2^t - \E_{X \sim D}[\iprod{v, X - \mu}^{t}]$ 
can be expressed as a sum of square polynomials of degree at most $d = O(t)$ in the entries of $v$), 
and $D$ has light tails. 
We highlight that our algorithmic technique can also be used in the dense case, 
under the same assumption, and provides qualitatively similar error guarantees 
with much simpler arguments and improved runtime. 
Below we apply our technique to the sparse setting.

Our results hold  when (i) the $t$-th moment of $D$ is $d$-certifiably bounded 
{for every $v$ that is $k$-sparse}, with $d = O(t)$, 
and (ii) $D$ has light tails.  For ease of exposition, 
we state the result assuming $D$ has subexponential tails 
(i.e., for some universal constant $c$, for all unit vectors $v$ 
and $p \in \N$, $\E_{X \sim D}[|\iprod{v, X-\mu}|^p]^{1/p} \leq cp$). \footnote{ 
	It is sufficient for $D$ to have bounded moments 
	up to degree $\poly(t \log(n))$ in the standard basis directions; 
	see \Cref{sec:prelims}.
}

\begin{restatable}[List-Decodable Sparse Mean Estimation]{theorem}{LDLsparse}\label{thm:LDL-sparse}
	Let $t$ be an integer power of two. 
	Let $D$ be a distribution over $\mathbb{R}^n$ 
	{with $k$-sparse mean $\mu$. }
	Suppose that $D$ has $t$-th moments $d$-certifiably bounded 
	in $k$-sparse directions by $M$ for some $d=O(t)$ (cf.\ \Cref{def:bounded-moments-k-sparse}) 
	and subexponential tails in the standard basis directions. There is  an algorithm which, 
	given $\alpha$, $M$, $t$, $k$, and a $(1{-}\alpha)$-corrupted set of 
	$m = (t k \log n)^{O(t)}~\max(1, M^{-2})/\alpha$ samples from $D$, 
	runs in time $\poly(mn^t)$ 
	and returns a vector $\hat{\mu}\in \R^n$ 
	such that with probability $\Omega(\alpha)$ 
	it is the case that
	$\norm{\hat{\mu} - \mu}_{2} = O_t(    M^{1/t}/\alpha^{O(1)/t})$.
\end{restatable}

Note that with high probability over the inliers and for any choice of outliers,  with probability $\Omega(\alpha)$ over the internal randomness of our algorithm, we find an estimate $\hat{\mu}$ close to $\mu$. By running our algorithm $O(1/\alpha)$ times, we can generate a list of size $O(1/\alpha)$ such that with probability $0.9$ the list contains the desired estimate $\hat{\mu}$.

Notably, for the important special case of Gaussian $\cN(\mu,I)$ inliers, 
our algorithm achieves the information-theoretically optimal error rate. 
This is because $\cN(\mu, I)$ has its $t$-th moment 
certifiably bounded by $t^{t/2}$ in all directions. 
Specifically, for a large enough constant $C > 0$, we obtain the following result:
Given $\alpha$, $t$, and a $(1- \alpha$)-corrupted set of 
$m \geq (t k \log n)^{Ct}$ samples from $\cN(\mu,I)$ for a $k$-sparse vector $\mu$, 
our algorithm runs in time $\poly(mn^t)$ 
and with probability $\Omega(\alpha)$ outputs a vector $\widehat{\mu}$ 
such that $\|\widehat{\mu} - \mu\|_2 \leq O(\sqrt{t}/\alpha^{C/t})$. 
Thus, by taking $t = C\log(1/\alpha)$, 
we obtain the optimal error of $\Theta (\sqrt{\log(1/\alpha)})$ 
in quasi-polynomial sample and time complexity.

We also note that a broad and natural class of distributions 
satisfying \Cref{def:bounded-moments-k-sparse} is the class of 
$\sigma$-Poincare distributions {(see, e.g., \cite{KStein17})}. 
A distribution is said to be $\sigma$-Poincare 
if for all differentiable functions 
$f:\R^n \rightarrow \R$, we have that 
$\Var_{X \sim D}\Brac{f(X)} \leq \sigma^2 \E_{X \sim D} [ \Norm{\nabla f(X)}_2^2]$. 

We complement our algorithm of \Cref{thm:LDL-sparse} 
with a qualitatively matching lower bound in the Statistical Query (SQ) model~\cite{Kearns:98}.
Instead of directly accessing samples, SQ algorithms are only allowed to perform 
adaptive queries of expectations of bounded functions of the underlying distribution, 
up to some desired tolerance (c.f.\ \Cref{def:stat}). 
The class of SQ algorithms is fairly broad: a wide range of known algorithmic techniques in
machine learning are known to be implementable in the SQ model (see, e.g.,~\cite{FGR+13}).

An SQ lower bound is an unconditional statement that { for any SQ algorithm,} 
either the number of queries $q$ must be large or the tolerance, $\tau$, of some query must be small. 
Since simulating a query of tolerance $\tau$ by averaging i.i.d.\ samples may need up to $\Omega(1/\tau^2)$ many of them, 
SQ lower bounds are naturally interpreted as a tradeoff between runtime $\Omega(q)$ 
and sample complexity $\Omega(1/\tau^2)$.
An adaptation of the result in~\cite{DKS18-list} 
yields \Cref{thm:sqlb-informal}, which indicates that the $k^{O(t)}$ factor 
in the sample complexity of \Cref{thm:LDL-sparse} 
might be necessary for efficient algorithms, {even for Gaussian inliers.}

\begin{theorem}[SQ Lower Bound, Informal]\label{thm:sqlb-informal}
	Consider the problem of list-decoding the mean of $\cN(\mu,I)$, 
	for a $k$-sparse vector $\mu \in \R^n$, up to error better than $O((t\alpha)^{-1/t})$. 
	Any SQ algorithm that solves the problem does one of the following: 
	(i) It returns a list of size $n^{ k^{\Omega(1)}}$, 
	(ii) it uses at least one query of tolerance $k^{-\Omega(t)} \exp({O(t\alpha)^{-2/t}})$, 
	or (iii) it {makes at least $n^{ k^{\Omega(1)}}$ queries.}
\end{theorem}

A similar lower bound holds for the computational model of low-degree polynomial tests, 
as a consequence of the recently established relationship between the two models \cite{BBHLS20}. 
See \Cref{sec:tradeoffs} for more details about the two models and the corresponding lower bounds.

\subsection{Overview of Techniques} \label{sec:our-tech}

We begin with a brief overview of the existing techniques for the \emph{dense} list-decodable mean estimation. We then highlight the obstacles in adapting these techniques to the sparse setting. Finally, we present an overview of our list-decodable mean estimation algorithm and how that overcomes these obstacles.

\paragraph{Prior Work} In the dense case, the existing techniques 
for performing list-decoding with error better than 
$\Omega(\alpha^{-1/2})$ are quite complicated. \cite{DKS18-list} uses a 
multifilter-based technique for list-decoding spherical Gaussians. 
This methodology relies critically on knowing the higher degree moments 
of the inliers (and thus does not generalize to less specific distribution families). Moreover, this method runs into technical difficulties 
related to being unable to determine the variance of higher degree polynomials 
on the inliers without knowing the mean ahead of time. 
The other approach in the literature (see, for example, \cite{KStein17,raghavendra2020list}) 
uses the Sum-of-Squares method (SoS) to find these clusters of points. 
 The algorithm in \cite{raghavendra2020list} involves solving a nested SoS program 
and then applying a complicated rounding procedure to get the final list. It should be noted that the runtime of~\cite{raghavendra2020list} is exponential in 
$\poly(1/\alpha)$, which can be quite large. 
Finally, \cite{KStein17} gave an SoS based list-decodable mean estimation 
algorithm for the \emph{dense} case with error, sample complexity, 
and list size similar to the ones that are obtained in our work; 
but significantly worse runtime.
The approach of \cite{KStein17} has some important differences. 
First, the clustering relaxation is conceptually harder, 
involving a more complex optimization problem for each filtering step 
and second, after the filtering ends, the error guarantee scales 
with the norm of the unknown mean; 
thus a complicated re-clustering step that combines ideas 
from \cite{SteinhardtCV18} is needed to reduce the error.

We present a significantly cleaner method to perform 
the outlier removal step; 
we avoid problems like not knowing the mean ahead of time by simply taking pairs of differences of our samples 
to make their mean zero. While it is likely that either 
	of the above techniques could be adapted to the sparse setting 
	with sufficient effort, this would result in significantly more complicated algorithms. We briefly point out some difficulties below.
	
	First, to ensure that the algorithms identify subsets 
	of the samples that satisfy certifiably bounded moments 
	in all $k$-sparse directions 
	requires additional variables and constraints to the algorithms. 
	Additionally, 
	one would need to replace the bounded moments 
	in all  directions condition 
	by the corresponding condition for the sparse case, 
	and ensure that all the proof steps 
	can be modified to rely only on the latter
	-- this would result in minor modifications 
	of the original algorithms, such as thresholding 
	of candidate solution vectors. 
	
	Second, at the end of this process, 
	while the algorithm might qualitatively 
	match the error guarantee that we achieve, 
	the runtime would continue to be $(mn)^{O(t)}$ for \cite{KStein17} or 
	$(1/\alpha)^{\polylog(1/\alpha)} n^{O(\max\{ 1/\alpha^4, t\})}$ for \cite{raghavendra2020list} -- 
	both of these are qualitatively worse than the runtime 
	we achieve when $\alpha$ is sufficiently small. 
	To obtain improved runtime using these prior techniques, one would require additional ideas, e.g., from \cite{steurer2021sos},  
	to be adapted to this setting, 
	overall resulting in a far more complicated algorithm. 
	
	On the other hand, it seems difficult to adapt the multi-filtering technique from \cite{DKS18-list} 
	to the setting we consider, without any introduction 
	of an SoS component. We remind the reader that 
	the  \cite{DKS18-list} algorithm depends critically
	on knowing the higher moments of the inliers \emph{exactly}, and does not generalize to less specific distribution families. Even in the Gaussian setting, generalizing~\cite{DKS18-list} might be difficult, since it would require the design of an efficiently verifiable notion of matching higher moments in $k$-sparse directions.

We believe that our novel list-decoding technique is significantly
simpler. As a result of simplifying the optimization programs involved,
our technique naturally improves the runtime from $\poly(m^tn^t)$ 
in prior work to $\poly(mn^t)$.

\paragraph{Novel List-Decoding Mean Estimation Algorithm} 
All known list-decoding algorithms are based upon the following fundamental observation. 
Suppose that $S$ is a set of samples that contains a subset $S_\text{g}$ of samples 
with bounded  moments. If we can find another subset $S'$ of $S$ with bounded  
moments and large overlap with $S_\text{g}$, then the means of $S'$ and $S_\text{g}$ 
cannot be too far apart (see \Cref{lem:clique_round_lemma}). 
The goal of a list-decoding algorithm is to find such a subset $S'$
(or, more precisely, a small number of hypotheses for such a subset). In order to generalize this to the sparse mean setting, 
it suffices for the subsets $S_\text{g}$ and $S'$ to have bounded moments only in all $k$-sparse directions. 
This will imply that $|\iprod{ v, \mu_\text{g} - \mu_{S'}}|$ is relatively small for all $k$-sparse 
unit vectors $v$, which will in turn imply that truncating $\mu_{S'}$ 
to its $k$ largest entries will provide a suitable approximation 
to $\mu_\text{g}$ (see \Cref{fact:sparseTruncation}).

The basic idea behind our novel list-decoding technique, which we call \emph{the difference of pairs filter}, is the following: let $T$ be the set of differences 
of pairs of elements of $S$. The set $T$ will contain a relatively large subset, 
$T_\text{g}$, (consisting of the pairwise differences of elements of 
$S_\text{g}$) whose moments in $k$-sparse directions are bounded. 
Our goal will be to find a subset $T' \subset T$ that has large overlap with 
$T_\text{g}$ and also has bounded ($k$-sparse) moments.

Na\"ively, we can do this as follows. 
We start with $T' = T$. Either this set has bounded $k$-sparse 
moments (in which case we are done) or there is some sparse direction 
$v$ in which the average value of $|\iprod{v,x}|^t$ over $T'$ is 
substantially larger than the average value over $T_\text{g}$. 
By throwing away points $x$ from $T'$ with probability proportional 
to $|\iprod{v,x}|^t$, we eliminate mostly bad points. 
We repeat this until $T'$ has bounded sparse moments. 
Unfortunately, while this approach can be shown to be correct, it does not suffice for our purposes because it is computationally infeasible
in general to determine whether or not $T'$ has bounded sparse moments. 
However, if we assume additionally that $S_\text{g}$ has bounded $k$-sparse 
moments \emph{provable by a low-degree sum-of-squares proof} 
(i.e., the moment bound inequality can be re-expressed 
as a sum of square polynomials being greater than $0$), 
then so will $T_\text{g}$. 
There is an efficient algorithm to see whether or not $T'$ has certifiably bounded moments in $k$-sparse directions as well.
If $T'$ does not have certifiably bounded moments, 
standard techniques for Sum-of-Squares programs imply 
that we can manufacture a non-negative polynomial $p$ 
so that the average value of $p$ on $T'$ is substantially 
larger than the value of $p$ on any set with 
SoS-certifiable bounded moments in $k$-sparse directions.
Thus, filtering out points $x$ in $T'$ with probability proportional 
to $p(x)$ will likely remove mostly bad points. 
Using this idea, we can find such a set $T'$ efficiently (see 
\Cref{thm:filter_list_decode_mean}).

We are thus left with a set of differences of samples 
rather than a set of samples. At this point, we will need a rounding method 
that given a set $T'$ of differences with bounded $k$-sparse moments 
guaranteed to have large overlap with $T_\text{g}$, 
finds a small list of sets $S'$ with bounded $k$-sparse moments 
so that at least one of them has a large overlap with $S_\text{g}$. Note that $T'$ might be the union of $T_i - T_i$, where each $T_i$ is an individual cluster drawn from a moment-bounded distribution. This demonstrates the necessity of a rounding step to finally identify the means of individual clusters.
To achieve this, it is helpful to think of $T'$ 
as consisting of a set of pairs of elements of $S$, 
or equivalently as a graph over $S$. Given $T'$, 
our first task is to find some reasonably large subset 
of $S$ with bounded   moments
(ideally which has a large overlap with $S_{\text{g}}$).
It is not hard to see that it suffices to find any large clique 
in $T'$ (see \Cref{lem:clique_round_lemma}). 
Unfortunately, $T'$ may well not have any large cliques, 
and even if it does, finding them may be computationally difficult.
However, we are saved here by the 
observation that if $|\iprod{v, x-y}|^t$ and $|\iprod{v, y-z}|^t$ 
are both small, then so is $|\iprod{v, x-z}|^t$. 
This means that if we replace $T'$ 
by the graph $H$ of all pairs of vertices $(v,w)$, 
where $v$ and $w$ have many common neighbors in $T'$, 
this new graph will \textbf{also} have relatively small 
$k$-sparse moments (see \Cref{lem:bounded_moments_transfers}). 
As most elements of $S_\text{g}$ are adjacent in $T'$ to most other elements of $S_\text{g}$, 
it is not hard to see that this new graph 
will have relatively large cliques; 
unfortunately, finding them may still be computationally difficult.

To find these large cliques efficiently, we need one final observation. 
If $v$ is a random vertex of a graph $G$, 
then there will (on average) be very few pairs of neighbors, 
$u$ and $w$, of $v$ so that $u$ and $w$ do not have 
a large number of common neighbors 
with each other in $G$ (see \Cref{lem:overlap_dense_subgraphs}). 
This means that if we pick a random sample $x$ in $S$, 
then most pairs of neighbors of $x$ in $T'$ are neighbors in $H$. 
Using a densification procedure (see \Cref{lem:overlap_prune_clique}), 
it is not hard to find a large subset $S'$ of these neighbors 
so that any two elements of $S'$ have many neighbors in common in $G$. 
By~\Cref{lem:bounded_moments_transfers}, 
this implies that $S'$ will in fact have bounded $k$-sparse moments.
Furthermore, if we happened to pick $x$ from $S_\text{g}$, 
it is not hard to see that it is likely that $S'$ 
has large overlap with $S_\text{g}$, and thus its mean 
provides us with a good estimate.

While the above describes our algorithmic approach, 
we also need to consider the sample complexity of our method. 
We know that the distribution $D$ has the property 
of bounded moments in $k$-sparse directions, 
and our algorithm requires that the property 
is also satisfied by the uniform distribution 
over the samples $S_\text{g}$. 
In order for $S_\text{g}$ to have it as well, it suffices 
that the $t^{th}$ moment tensor of $S_\text{g}-\mu$ 
be close to the corresponding moment tensor of $D-\mu$. 
It turns out that if these moment tensors 
are $\delta k^{-t}$-close coordinate-wise, 
which takes $O(t \log(n) k^{O(t)}/\delta^2)$ samples, 
this suffices to get the kind of certifiable concentration we require. 
This is captured by \Cref{lem:pop-to-empirical}, a restatement 
from~\cite{DKKPP22}.

\section{Preliminaries} \label{sec:prelims}

\paragraph{Basic Notation} 
 We use $\N$ to denote natural numbers and $\Z_+$ to denote positive integers. For $n \in \Z_+$ we  denote $[n] := \{1,\ldots,n\}$. 
We denote by $\R[x_1,\ldots, x_n]_{\leq d}$ the class of real-valued polynomials of degree at most $d$ in variables $x_1,\ldots, x_n$.
 We use $\poly(\cdot)$ to indicate a quantity that is polynomial in its arguments. 
 For an ordered set of variables $Q = \{x_1, \dots, x_n\}$, we will denote $p(Q)$ to mean $p(x_1, \dots, x_n)$. Throughout the paper, we will typically use the letter $n$ for the dimension, $m$ for the number of samples, $d$ for the degrees of the SoS proofs, and $t$ for the number of bounded moments.
 
\paragraph{Linear Algebra Notation} We use $I_n$ to denote the $n\times n$ identity matrix. We will drop the subscript when it is clear from the context.
We typically use small case letters for deterministic vectors and scalars. We will specify the dimensionality unless it is clear from the context. 
We denote by $e_1,\ldots,e_n$ the vectors of the standard orthonormal basis, i.e., the $j$-th coordinate of $e_i$ is equal to $\mathbf{1}_{\{i=j \}}$, for $i,j \in [n]$.
 For a vector $v$, we let $\|v\|_2$ denote its $\ell_2$-norm. We call a vector $k$-sparse if it has at most $k$ non-zero coordinates. 
 We use $\langle v,u \rangle$ for the inner product of the vectors $u,v$.
We will use $\cdot ^{\otimes s}$ to denote the standard Kronecker product.

\paragraph{Probability Notation} We use capital letters for random variables. For a random variable $X$, we use $\E[X]$ for its expectation.	
We use $\cN(\mu,\Sigma)$ to denote the Gaussian distribution with mean $\mu$ and covariance matrix $\Sigma$. We let $\phi$ denote the pdf of the one-dimensional standard Gaussian.
When $D$ is a distribution, we use $X \sim D$ to denote that the random variable $X$ is distributed according to $D$. When $S$ is a set, we let $\E_{X \sim S}[\cdot]$ denote the expectation under the uniform distribution over $S$. For any sequence $a_1, \dots, a_m \in \mathbb{R}^n$, we will also use $\E_{i \sim [m]}[a_i]$ to denote $\frac 1 m \sum_{i \in [m]} a_i$. 
For a real-valued random variable $X$ and $p \geq 1$, we use $\|X\|_{L_p}$ to denotes its $L_p$ norm, i.e., $\|X\|_{L_p} := (\E[|X|^p])^{1/p}$.
For a unit vector $v \in \R^d$ and a distribution $P$ over $\R^d$ with mean $\mu$, we define the $i$-th moment of $P$ in the direction $v$ to be $\E[|\langle v, X - \mu \rangle|^i ]$.

\medskip

\begin{restatable}[(2, $k$)-norm]{definition}{TWOKNORM}
We define the $(2, k)$-norm of a vector $x$, denoted as $\|x\|_{2,k}$, to be the maximum correlation with any $k$-sparse unit vector, i.e., 
$\Norm{x}_{2,k} := \max_{\Norm{v}_2=1, v: k\mathrm{-sparse}} \iprod{v, x}$. 
\end{restatable}

The following standard fact translates bounds from the $(2,k)$-norm to the usual $\ell_2$-norm when the underlying mean $\mu$ is  $k$-sparse (see, e.g., \cite{DKKPP22} for a proof):
\begin{restatable}{fact}{FactTruncSparse} 
	\label{fact:sparseTruncation}
	Let $h_k: \R^n \to \R^n$ denote the function where $h_k(x)$ is defined to truncate $x$ to its $k$ largest coordinates in magnitude and zero out the rest. For all $\mu \in \R^n$ that are $k$-sparse, we have that $\|h_k(x) - \mu\|_2 \leq 3 \|x-\mu\|_{2,k}$.
\end{restatable}
 
\subsection{SoS Preliminaries}
The following notation and preliminaries are specific to the SoS part of this paper.  We refer to \cite{BarakSteurerNotes} for a more complete treatment of the SoS framework. Here, we review the basics.

\begin{definition}[Symbolic Polynomial]
	A degree-$d$ symbolic polynomial $p$ is a collection of indeterminates $\widehat{p}(\alpha)$, 
	one for each multiset $\alpha \subseteq [n]$ of size at most $d$.
	We think of it as representing a polynomial $p \, : \, \R^n \rightarrow \R$ whose coefficients are themselves 
	indeterminates via $p(x) = \sum_{\alpha \subseteq [n], |\alpha| \leq t} \widehat{p}(\alpha) x^\alpha$.
\end{definition}

\begin{definition}[SoS Proof]\label{def:sos-proof}
	Let $x_1,\ldots,x_n$ be indeterminates and let $\cA$ be a set of polynomial equalities $\{ p_1(x) = 0,\ldots,p_m(x) = 0 \}$.
	An SoS proof of the inequality $r(x) \geq 0$ from axioms $\cA$ is a set of polynomials $\{r_i(x)\}_{i \in [m]} \cup \{r_0(x)\}$ such that $r_0$ is a sum of square polynomial and $r_{i}$'s are arbitrary, and $r(x) = r_0(x) + \sum_{i \in [m]} r_i(x)  p_i(x).$
	If the set of polynomials $\{ r_i(x) \cdot p_i(x) \mid i \in [d]\}\cup \{r_0(x)\}$ have degree at most $d$, we say that this proof is of degree $d$ and denote it by $\cA \sststile{d}{} r(x) \geq 0$.  
	When we need to emphasize what indeterminates are involved in a particular SoS proof, we denote it by  $\cA \sststile{d}{x} r(x) \geq 0$.
	When $\cA$ is empty, we omit it, e.g., $\sststile{d}{} r(x) \geq 0$ or $\sststile{d}{x} r(x) \geq 0$.
\end{definition}

We will also use the objects called \emph{pseudoexpectations}.
\begin{definition}[Pseudoexpectation]
	Let $x_1,\ldots,x_n$ be indeterminates.
	A degree-$d$ pseudoexpectation $\pE$ is a linear map $\pE  : \R[x_1,\ldots,x_n]_{\leq d} \rightarrow \R$
	from degree-$d$ polynomials to $\R$ such that $\pE \Brac{p(x)^2} \geq 0$ for any $p$ of degree at most $d/2$ and $\pE \Brac{1} = 1$.
	If $\cA = \{p_1(x) = 0, \ldots, p_m(x) = 0\}$ is a set of polynomial inequalities, we say that a pseudoexpectation $\pE$ satisfies $\cA$ if for every $i \in [m]$,
$\pE [s(x)  p_i(x)] = 0$  
	for all polynomials $s(x)$ such that $s(x)  p_i(x)$ has degree at most $d$. 
\end{definition}

It is well known (see, e.g.,~\cite{BarakSteurerNotes}) that pseudoexpectations are dual objects to SoS proofs in the following sense: given a set $\cP$ of $r$ polynomial  equalities in $n$ variables and a polynomial $q(x_1,\ldots,x_n)$, either there exists an SoS proof $\cP \sststile{\ell}{x} q(x)\geq 0$  or there exists a pseudoexpectation $\pE$ of degree $\ell$ satisfying $\cP$ but having $\pE[q(x)] < 0$. More importantly, there is an algorithm that runs
in time $(rn)^{O(\ell)}$ and finds that pseudoexpectation when we are in the second case.

\begin{theorem}[The SoS Algorithm \cite{Sho87,lasserre2001new, nesterov2000squared,bomze1998standard}] \label{thm:sos-algo}
For any $n,r, \ell \in \Z^+$, and a set of $r$ polynomial equalities $\cP=\{p_1(x_1,\ldots, x_n)=0,\ldots, p_r(x_1,\ldots, x_n)=0\}$, the following set has an $(rn)^{O(\ell)}$-time weak separation oracle (in the sense of~\cite{GLS81}):
\begin{align*}
 \{ q(x_1,\ldots, x_n) : \cP \sststile{\ell}{x_1,\ldots, x_n}  q(x_1,\ldots, x_n) \geq 0 \}   
\end{align*}
\end{theorem}

Finally, it is standard fact that several commonly used inequalities like the triangle inequality or Cauchy-Schwartz inequality have an SoS version.

\begin{fact}[SoS Cauchy-Schwartz and H\"older (see, e.g., \cite{hopkins2018clustering})]\label{fact:sos-holder}
	Let $f_1,g_1,  \ldots, f_n, g_n$ be indeterminates.
	Then, 
	\begin{align*}
	\sststile{2}{f_1, \ldots, f_n,g_1, \ldots, g_n} \Set{ \Paren{\frac{1}{n} \sum_{i=1}^n f_i g_i }^{2} \leq \Paren{\frac{1}{n} \sum_{i=1}^n f_i^2} \Paren{\frac{1}{n} \sum_{i=1}^n g_i^2} } \;.
	\end{align*} 
\end{fact}
\begin{fact}[SoS Triangle Inequality]\label{fact:sos-triangle}
	If $k$ is a power of two, 
	$\sststile{k}{a_1, a_2, \ldots, a_n} \Set{ \left(\sum_i a_i \right)^k \leq n^k \Paren{\sum_i a_i^k} }.$ 
\end{fact}

We also require the following fact, which permits us to take square roots on both sides of an inequality, provided that the right-hand side is assuredly positive.

\begin{fact}[SoS square root, Lemma A.3 from \cite{KS17} specialized]\label{prop:unsquare}
    Let $\cA = \{p_1,\ldots,p_m \}$ be a set of axioms, all of which are polynomials in the variable $x$, let $q(x)$ be another polynomial of degree $t$ in the variable $x$, and let $M$ be a positive real number which is independent of $x$. If $\cA \sststile{d}{x} M^2 - q^2(x) \geq 0$ then $\cA \sststile{d+2t}{x} M- q(x) \geq 0$.
\end{fact}

\subsection{Certifiably Bounded Moments in $k$-Sparse Directions}
Our algorithm succeeds whenever the uncorrupted samples have \emph{certifiably bounded moments} in $k$-sparse directions, defined as in~\cite{DKKPP22}: 
\begin{definition}[$(M,t,d)$-Certifiably Bounded Moments in $k$-Sparse Directions]
	\label[definition]{def:bounded-moments-k-sparse}
	Let $Q := \{v_1,\dots$ $, v_n, z_1, \dots, z_n\}$ and $\ak$ $:= \{ z_i^2 = z_i\}_{i \in [n]}$ $ \cup \{v_i z_i = v_i\}_{i \in [n]}$ $\cup \Set{ \sum_{i=1}^n z_i = k }  \cup \Set{\sum_{i=1}^n v_i^2 = 1 }.$ For an $M>0 $ and even  $t \in \N$, a distribution $D$ with mean $\mu$ satisfies \emph{$(M,t,d)$ certifiably bounded moments in $k$-sparse directions}  
	if 
	\begin{align*}
	\ak \sststile{d}{Q} \E_{X \sim D}\Brac{\iprod{ v,  X - \mu}^{t} } \leq M \;.
	\end{align*}
\end{definition}
The definition of $\ak$ is based on the fact that a vector $v= (v_1,\dots,v_n)$ is $k$-sparse if and only if there exists $z = (z_1,\dots,z_n)$ such that $v,z$ satisfy $\ak$.

We will use the following lemma proved in~\cite{DKKPP22} to bound the number of samples it takes to certify bounded moments in $k$-sparse directions. Although this is stated for subexponential distributions, it  requires only that  the distribution has bounded $t^2 \log(n)$ moments in the standard basis directions (see \Cref{lem:basic_linf_consc-full}).
The lemma proves a bound for the square of the moment, but, given \Cref{prop:unsquare}, this implies the bound without the squares, as used in \Cref{def:bounded-moments-k-sparse}.

\begin{restatable}[\cite{DKKPP22}]{lemma}{SamplingPreserves}
\label[lemma]{lem:pop-to-empirical}
	Let $D$ be a distribution over $\mathbb{R}^n$ with mean $\mu$. Suppose that $D$ has $c$-sub-exponential tails in the standard basis directions around $\mu$ for a constant $c$ and that  $\ak \sststile{O(t)}{Q} \E_{X \sim D} \Brac{\iprod{v, X-\mu}^t}^2 \leq M^2$. 
Let $S= \{X_1, \ldots, X_m\}$ be a set of $m$ i.i.d.\ samples from $D$, $D'$ be the uniform distribution over $S$, and $\overline{\mu}:=\E_{X \sim D'}[X]$.
	If $m > C (t k (\log n))^{5 t} \max(1,M^{-2})$ for a sufficiently large constant $C$, then with probability at least $0.9$ we have the following: 
	\begin{enumerate}
	    \item $\ak \sststile{O(t)}{Q} \E_{X \sim D'} \Brac{\iprod{v, X-\overline{\mu}}^t}^2 \leq 8M^2$.
	    \item $\iprod{ v,\overline{\mu} -\mu } \leq {M^{1/t}}/{\alpha^{6/t}}$ for every $k$-sparse unit vector $v$.
	\end{enumerate}
\end{restatable}

\section{Main Result: Proof of \Cref{thm:LDL-sparse}} \label[section]{sec:list-decodable-sparse-mean-estimation}
Recall our setting: we are given $\alpha \in (0, 1/2)$ and a 
multiset $S:= \{x_1, \dots, x_m \}$ 
such that an unknown subset of $\lfloor \alpha m \rfloor$ many of these points
satisfy $(M, t, d)$-certifiably bounded moments 
in $k$-sparse directions, and the remaining are arbitrary. 
The goal is to recover a candidate vector that is close 
to $\mu := \E_{X \sim D}[X]$ with probability $\Omega(\alpha)$. 
In what follows, $t$ will always be $2^\ell$ for some $\ell \in \Z_+$. 

\subsection{The SoS-based Filter}

Let $T$ be the set of \emph{pairwise differences} of all samples in $S$ 
(similarly denote by $T_{\text{g}}$
the subset of $T$ corresponding to inliers $S_{\text{g}}$). 
We would like to either detect that
the moments of $T$ are already certifiably bounded in all $k$-sparse directions 
or find a direction that violates this and filter out mostly outliers. 
Unfortunately, this kind of check is computationally infeasible, 
but it can be done efficiently if we check for moment bounds that are certified by
SoS proofs. 
This is done in \Cref{alg:sos_filter_list_decode}, which takes as input the set $T$ along with
the parameters $t,d,M$ and performs filtering until
the resulting set $T'$ has $t$-th moments bounded by $M$.

\begin{algorithm}
	\begin{algorithmic}[1]
		\Function{LDMean-SoS-Filter}{$T := \{x_1, \dots, x_m\},t, d, M$}
		\State Let $Q = \{v_1, \dots, v_n, z_1, \dots, z_n \}$ and $T'=T$.
		\While{there is no SoS proof of $\ak \sststile{d}{Q} \sum_{x \in T'} \iprod{v, x}^t \leq 6M|T|$} \label{line:pe}
		\State Find a  degree-$d$  $\pE$ on $Q$ satisfying $\ak$ and $\pE[\sum_{x \in T} \iprod{v, x}^t] > 6M|T|$. \label{line:pe2}
		\State Throw out $x \in T'$ with probability $\pE \Brac{\iprod{v, x}^t}/ \max_{x\in T'}\pE \Brac{\iprod{v, x}^t}$.
		\EndWhile 
		\State \Return $T'$
		\EndFunction
	\end{algorithmic}
	\caption{SoS-based filter for list-decodable mean estimation}
	\label{alg:sos_filter_list_decode}
\end{algorithm}
\noindent

\begin{theorem}[Filter Identifies a Subset Satisfying Bounded Moments] \label{thm:filter_list_decode_mean}
Let $T$ be a multiset of points in $\R^n$ for which there exists a subset $T_{\text{g}} \subset T$ with
	$|T_{\text{g}}| = \alpha^2 |T|$ for some $\alpha  > 0$. Furthermore assume that $T_{\text{g}}$ has zero mean and $(M,t,d)$-certifiably bounded
	moments in
	$k$-sparse directions for some $M>0, d \in \Z_+$, and even $t$. Then \Cref{alg:sos_filter_list_decode}, 
	given $T, M,t,d$, returns a subset $T' \subseteq T$ in time $\poly(mn^d)$ so that, with probability at least $2/3$:
	\begin{enumerate}
		\item For any $k$-sparse unit vector $v$, we have $\sum_{x \in T'} \iprod{v, x}^t \leq 6M|T|$.
		\item $|T'\cap T_{\text{g}}| \geq |T_{\text{g}}|/2$.
	\end{enumerate}
\end{theorem}
\begin{proof}
	Let $Q = \{z_1, \dots, z_n, v_1, \dots, v_n\}$. 
In \Cref{alg:sos_filter_list_decode}, \Cref{line:pe,line:pe2} use the separation oracle of \Cref{thm:sos-algo} with $\cP=\ak$.
	Since $T_{\text{g}}$ has zero mean and $(M,t,d)$-certifiably bounded moments in $k$-sparse directions, we have that
	\begin{align} \ak \sststile{d}{Q} M - \E_{X \sim T_{\text{g}}}\Brac{\iprod{v, X}^{t}} \geq 0. \label{eqn:S_good-bounded-moments}
	\end{align}
	Thus, if $\ak \sststile{d}{Q} \sum_{x \in T'} \iprod{v, x}^t \leq 6M|T|$ 
	the algorithm identifies that using the separation oracle  and stops (in which case we have the 
	desired conclusion that for any $k$-sparse unit vector 
	$v$,   $\sum_{x \in T'} \iprod{v, x}^t \leq 6M|T|$). Otherwise, the separation oracle returns a degree-$d$ pseudo-expectation $\pE$
	on $Q$ satisfying $\ak$ and 
	\begin{align}\label{eq:psexp_ineq}
		\pE\Brac{-6M|T| + \sum_{x \in T'} \iprod{v, x}^t} > 0 \;,
	\end{align}
	in which case we can create a filter:  
	Using  \eqref{eqn:S_good-bounded-moments} and  
	$|T_{\text{g}}| = \alpha^2 |T|$, we have that 
	\begin{align*}
		M \geq\E_{X \sim T_{\text{g}}}\Brac{ \pE{[\iprod{v, X}^{t}]}} \geq 
		\frac{\alpha^{-2}}{|T|} \sum_{x \in T_{\text{g}} \cap T'} \pE 
		\Brac{\iprod{v, x}^{t}} \;. 
	\end{align*}
	By \eqref{eq:psexp_ineq} and linearity of $\pE$, we see that
	\begin{align*}
		\frac{\sum_{x \in T_{\text{g}} \cap T'} \pE \Brac{\iprod{v, x}^t} }{\sum_{x \in T'} \pE \Brac{\iprod{v, x}^t}} \leq \alpha^2/6 \;.
	\end{align*} 
	This means if we throw out each sample $x$ with probability 
	$\pE \Brac{\iprod{v, x}^t}/ \max_{x\in T'}\pE \Brac{\iprod{v, x}^t}$ 
	(which is indeed in $[0,1]$ since $\iprod{v,x}^t$ 
	is a square, so its pseudoexpectation is a non-negative value), 
	on average, only an $\alpha^2/6$ fraction of the points 
	that are removed will be from $T_{\text{g}}$. 
	Since the sample $x$ with the largest value of $\pE\Brac{\iprod{v,x}^t}$ 
	will always be removed, the algorithm will terminate in polynomial time.

	We now analyze the size of the set $T' \cap T_{\text{g}}$ across the iterations. 
	By the above analysis, at each step, the expected number 
	of samples thrown out from $T_{\text{g}}$ is at most $\alpha^2/6$ 
	times the expected total number of samples removed. 
	{Thus, the potential function 
		\begin{align*}
			\Delta := \frac{|T_{\text{g}} \cap T'| - (\alpha^2/6)|T'|}{|T|}
		\end{align*}
		is a submartingale.
		By definition, $0 \leq \Delta/\alpha^2 \leq 1$ holds always, 
		and initially we had $\Delta/\alpha^2 \geq 5/6$. 
		By  Doob's martingale inequality (\Cref{fact:markov_martingales} 
		applied with $t = 1/2$ to the submartingale $\Delta/\alpha^2$), 
		the probability that $\Delta/\alpha^2$ remains at least 
		$1/2$ throughout the execution of the algorithm 
		is at least $2/3$. Thus, we will have 
		$|T_{\text{g}} \cap T'| \geq (\alpha^2/2)|T| = |T_{\text{g}}| / 2$ throughout the execution.}
\end{proof}

\subsection{Identifying a Subset of Samples with Bounded Moments}

Having identified a subset $T' \subset T$ satisfying the conclusions of \Cref{thm:filter_list_decode_mean}, 
we want to extract from $T'$ a vector that is close to the original mean. 
Since the average of the set of differences is likely to be close to zero regardless of the true mean, 
we will need to use the information about the pairs that we get from $T'$ to find subsets of the original 
samples that satisfy the appropriate concentration bounds. We will need the following definition.
\begin{definition}\label[definition]{def:graph_sparse_bounded_moments}
	Let $S \subset \mathbb{R}^n$. A graph $(V, E)$ on $S$ with $V = S$ is said to have $(M,t)${-bounded} moments in $k$-sparse directions if for all $k$-sparse unit vectors $v$,
	\begin{align*}
		\frac{1}{|S|^2}\sum_{(x, y) \in E} \iprod{v, x-y}^t \leq M \;.
	\end{align*}
\end{definition}
{By the guarantee of our filter, if $T'$ is the set returned by \Cref{alg:sos_filter_list_decode},  the graph $G$ with edges $(x, y)$ for which $x-y$ or $y-x$ belongs to $T'$ will have bounded moments in the sense of \Cref{def:graph_sparse_bounded_moments}.}
If $G$ contains a clique $C$ which intersects with an $\alpha$-fraction of the target samples $C_{\text{g}}$, the following result shows that the means of $C$ and $C_{\text{g}}$ are close. 

\begin{lemma}\label{lem:clique_round_lemma}
	Let $S \subset \R^n$ and $G$ be a graph on $S$ satisfying \Cref{def:graph_sparse_bounded_moments}. 
	Let $C \subset S$ be a clique in $G$. Let $C_{\text{g}} \subset C$ be a subset with
	$|C_{\text{g}}| \geq \alpha |S|$. If $\mu_C$ and $\mu_g$ denote the means 
	of $C$ and $C_{\text{g}}$ respectively, 
	{then $\iprod{v,  \mu_C - \mu_g}^t \leq  {2M}/{\alpha^2}$ for all $k$-sparse unit vectors $v$.}
\end{lemma}
\begin{proof} By using the fact that $t$ is even and the fact that $G$ satisfies \Cref{def:graph_sparse_bounded_moments}, we see that
	\begin{align*}
		M|S|^2 \geq  \sum_{(x, y) \in E} \iprod{v, x-y}^t \geq \frac{1}{2} \sum_{x, y \in C} \iprod{v, x-y}^t \geq \frac{1}{2} \sum_{x \in C_{\text{g}}, y \in C} \iprod{v, x-y}^t \;.    
	\end{align*}
	Using Jensen's inequality we obtain, 
	\begin{align*}
		M|S|^2 {\geq}   \sum_{x \in C_{\text{g}}, y \in C} \frac{\iprod{v, x-y}^t}{2}	{\geq} \frac{|C_{\text{g}}||C| \iprod{v, \mu_C - \mu_g}^t}{2}  {\geq} \frac{|S|^2 \alpha^2 \iprod{v, \mu_C - \mu_g}^t}{2} \;.
	\end{align*}
\qedhere
\end{proof}

Unfortunately, even the inliers might not form a clique in $G$. However, the guarantee that 
$|T' \cap T_g|\geq |T_g|/2$ implies that the inliers share many neighbors in the graph $G$. 
Thus we look at the overlap graph defined below, in the hope that this graph will be more dense.

\begin{definition}[Overlap Graph]
	Let $G = (V, E)$ be a graph and $\gamma > 0$. The \emph{overlap graph} $R_\gamma(G)$ is defined to be the graph with the vertex set $V$ where each $(x, y)$ is an edge in the graph iff $|N_G(x) \cap N_G(y)| \geq \gamma |V|$, where $N_G(x)$ denotes the neighborhood of the vertex $x$ in $G$. 
\end{definition}

The following result shows that if $G$ has bounded moments, then so does $R_\gamma (G)$.

\begin{lemma}[If $G$ Has Bounded Moments, Then $R_\gamma (G)$ Has Bounded Moments]\label{lem:bounded_moments_transfers}
	Let $S$ be a set of points and $G$ be a graph with $(M, t)$-bounded
	moments in $k$-sparse directions. 
	Then for $\gamma > 0$, $R_\gamma(G)$ has $(2 \cdot 2^t M/\gamma, t )$-bounded 
	moments in $k$-sparse directions. 
\end{lemma}
\begin{proof} 
Let $v$ be any arbitrary $k$-sparse unit vector.
	For any $x, y$ in $R_\gamma (G)$, the  triangle inequality implies \[ \iprod{v, (x-a) - (a-y)}^t \leq 2^t [\iprod{v, x-a}^t + \iprod{v, y-a}^t]. \] By taking a sum over all $a$ in $N_G(x) \cap N_G(y)$, we have
	\begin{align*}
		\iprod{v, x-y}^t &= \hspace{-20pt}\sum_{a \in N_G(x) \cap N_G(y)} \hspace{-4pt} \frac{\iprod{v, (x-a) - (a-y)}^t}{|N_G(x) \cap N_G(y)|} \leq (2^t/{\gamma |S|}) \hspace{-20pt} \sum_{a \in N_G(x) \cap N_G(y)}  [\iprod{v, x-a}^t + \iprod{v, y-a}^t ] \;. 
	\end{align*}
	Denote $\Gamma(\alpha,x) := \{y: \text{ neighbor of $x$ in }R_\gamma(G) \text{ and neighbor of $a$ in $G$}\} $. 
	Summing over all the edges $(x,y)$ in $R_\gamma(G)$, we have that
	\begin{align*}
		\sum_{(x,y) \in R_\gamma(G)} \iprod{v, x-y}^t &\leq \frac{2^t}{\gamma |S|} \sum_{(x,y) \in R_\gamma(G)} \sum_{a \in N_G(x) \cap N_G(y)}[\iprod{v, x-a}^t + \iprod{v, y-a}^t ] \\
		&= \frac{2 \cdot 2^t}{\gamma |S|} \sum_{(a,x) \in E} \sum_{y \in \Gamma(\alpha,x)} \iprod{v, x-a}^t \\
		&\leq \frac{2 \cdot 2^t \cdot |S|}{\gamma |S|} \sum_{(a, x) \in E} \iprod{v, a-x}^t \\
		&\leq  \frac{2 \cdot 2^t M|S|^2}{\gamma},
	\end{align*}
	where we use that $t$ is even and  $G$ has $(M,t)$-bounded moments in $k$-sparse directions.
\end{proof} 

While $R_\gamma (G)$ may not have any cliques either, it is guaranteed to have fairly dense subgraphs. 

\begin{lemma}[$R_\gamma (G)$ Has Dense Subgraphs]\label{lem:overlap_dense_subgraphs}
	Let $G = (V, E)$ be a graph and $\gamma >0$. If $x$ is a randomly selected vertex of $G$, then the expected number of pairs $y, z \in N_G(x)$ 	so that $y$ and $z$ are not neighbors in $R_\gamma(G)$ is at most $\gamma |V|^2$. 
\end{lemma}
\begin{proof} 
	
	The expectation in question is $1/|V|$ times the number of triples $x, y, z \in V$ so that $y$ and
	$z$ are not neighbors in $R_\gamma (G)$, but are both neighbors of $x$ in $G$. By the definition of $R_\gamma (G)$,
	if $y$ and $z$ are not neighbors in $R_\gamma (G)$, they have at most $\gamma |V|$ common neighbors in $G$. Thus, the number of such triples is at most $\gamma |V|^3$, so the expectation in question is at most $\gamma |V|^2$. 
\end{proof} 
{As outlined above, the inliers in $R_\gamma(G)$ form a dense subgraph. The next procedure (\textsc{Pruning} in \Cref{alg:clique_create})} prunes out points from a dense subgraph (inliers in $R_\gamma (G)$ for us) to find a clique.

\begin{lemma}[Dense Subgraphs Can Be Pruned to Obtain a Clique]\label{lem:overlap_prune_clique}
	Let $G = (V, E)$ be a graph and let $W \subset V$ be a set of vertices with $|W| = \beta |V|$
	and all but $\gamma |V|^2$ pairs of vertices in $W$ are connected in $G$, for $\beta, \gamma > 0$ 
	with $\gamma \leq \beta^2/36$. There exists an algorithm (\textsc{Pruning} in \Cref{alg:clique_create})  
	that given $G, W, \beta, \gamma$ runs in polynomial time and returns a  $W' \subset W$ so that 
	$|W'| \geq |W| - (6\gamma/\beta) |V|$ and so that $|W'|$ is a clique in $R_{\beta/3}(G)$. 
\end{lemma}

\begin{algorithm}[h]
	\begin{algorithmic}
		\Function{{Pruning}}{$G = (V, E), W \subset V$}
		\State Let $W' = W$
		\State  \textbf{while} {$\exists x \in W'$ that is not connected to at least $2|W|/3$ vertices in $W'$}:  Remove $x$ from $W'$ \label{line:whileloop}
		\State \Return $W'$
		\EndFunction 
	\end{algorithmic}
	\caption{Algorithm for clique creation}
	\label{alg:clique_create}
\end{algorithm}

\begin{proof} 
	In \Cref{line:whileloop} of the Algorithm, the point $x$ which is removed satisfies 
	$|N_G(x) \cap W'| < 2/3 |W|$. If we also have that $|W'| \geq 5|W|/6$ (something that we will verify later),
	the removal of $x$ decreases the number of pairs of unconnected elements in $W'$ by 
	$|W'| - |N_G(x) \cap W'| \geq |W'|-(2/3)|W|\geq  |W|/6 = (\beta/6)|V|$.
	This can happen at most $(6\gamma/\beta) |V|$ times before we run out of unconnected pairs of elements in $W'$, 
	thus $|W'| \geq |W| - (6\gamma/\beta) |V|$ upon termination.
	Also, since it holds $(6\gamma/\beta) |V| \leq (\beta/6) |V| = |W|/6$, 
	we indeed have $|W'| \geq 5|W|/6$ as claimed at the start.
	Now note that each element of $W'$	is connected to at least $2|W|/3$ other elements of $W'$ in $G$. 
	Thus any pair of elements of $W'$ have at least $|W|/3$ common neighbors,
	and thus are adjacent in $R_{\beta/3}(G)$.
\end{proof}

We are finally ready to prove our main algorithmic result on rounding. 
{The basic idea is that most of the inliers in $G$ (which are at least $\alpha$-fraction of vertices) 
are connected to many other inliers, and thus if we start with an inlier, its neighborhood will also 
contain many inliers and will be dense in the overlap graph $G'$ (\Cref{lem:overlap_dense_subgraphs}). 
	Thus, we can apply the pruning of \Cref{lem:overlap_prune_clique} 
	to obtain a large clique in the overlap graph of $G'$, which also has bounded moments by two 
	applications of  \Cref{lem:bounded_moments_transfers}. } 
\begin{theorem}[Rounding]\label{thm:rounding_list_decode}
	Let $S \subset \R^n$ and let $G=(V,E)$ be a graph with $V=S$ and $(M, t)$-bounded moments in $k$-sparse
	directions. Suppose there is a subset $S_{\text{g}} \subset S$ with $|S_{\text{g}}| \geq \alpha |S|$ and at
	least half of the pairs of points in $S_{\text{g}}$ are connected by an edge in $G$. Suppose that the
	$S_{\text{g}}$ has mean $\mu_{\text{g}}$ and $t$-th moment bounded by $M$ in $k$ sparse directions. 
	Then, there exists a randomized algorithm that given $G, S$ and $\alpha$ runs in polynomial time and returns a
	$\widehat{\mu} \in \R^n$ such that with probability
	$\Omega(\alpha)$, {for all $k$-sparse unit vectors $v$, }$ \iprod{ v, \widehat \mu - \mu_{\text{g}} }^t = O(10^tM \alpha^{-6} )$.
\end{theorem}

\begin{algorithm}{}
	\begin{algorithmic}[]
		\Function{Rounding}{$S, G = (V, E)$}
		\State Let $\delta = \alpha^3/4608$.
		\State Choose $x \in S$ uniformly at random, let $W = N_G(x)$, and let $G' = R_\delta(G)$\label{line:choosex}
		\State \textbf{if} {the number of pairs of points in $W$ that are not connected in $G'$ is more than $(8\delta/\alpha) |V|^2$ or if $|W| \leq (\alpha/4) |V|$} \Return FAIL \label{line:cond}
		\State \textbf{else}  Run \textsc{Pruning} on $G'$ and $W$ to obtain $W'$
		\State \Return $\E_{X \sim W'}[X]$. 
		\EndFunction
	\end{algorithmic}
	\caption{Algorithm for rounding}
	\label{alg:rounding}
\end{algorithm}

\begin{proof}
	The algorithm we consider is \textsc{Rounding} (\Cref{alg:rounding}). Let $x$ and $G'$ be as in \Cref{line:choosex}. We will claim that algorithm \textsc{Rounding} succeeds as long as the following hold:
	\begin{enumerate}
		\item $x \in S_{\text{g}}$,
		\item  $x$ has at least $S_{\text{g}}/4$ neighbors in $S_{\text{g}}$ in $G'$, and
		\item the number of pairs of neighbors of $x$ that are not neighbors in $G'$ is at most $(8\delta/\alpha)|V|^2$.
	\end{enumerate}
	First, we show that these conditions hold with probability $\Omega(\alpha)$. The first condition holds with
	probability at least $\alpha$ over the choice of $x$. Conditioned on $x \in S_{\text{g}}$, the expected number
	of non-neighbors that $x$ has in $S_{\text{g}}$ is at most
	$|S_{\text{g}}|/2$. Thus, the probability that it has more than $3|S_{\text{g}}|/4$ non-neighbors is at most $2/3$ by
	Markov's inequality. Thus, the first two conditions both
	hold with probability at least $\alpha/3$. Finally, the expected number of pairs of neighbors of $x$ that are
	non-neighbors in $G'$ is at most $\delta |V|^2$ by
	\Cref{lem:overlap_dense_subgraphs}. Thus, by Markov's inequality, there will be more than $(8\delta/\alpha)
	|V|^2$ such non-connected neighbors with probability at
	most $\alpha/8$. Combining with the	above, all three conditions hold with probability at least $\alpha/24$.

	Given these assumptions, we note that $|W| \geq |S_{\text{g}}|/4 \geq (\alpha/4) |V|$, and at most
	$(8\delta/\alpha) |V|^2$ of pairs in $W$ are not connected in $G'$. This implies that we pass the condition in
	\Cref{line:cond}.
	{We will now verify the conditions in \Cref{lem:overlap_prune_clique}. Since $\beta := |W|/|V| \geq \alpha/4$ and $\gamma$, the number of pairs of
	vertices in $W$ that are not connected in $G'$, is at most $8 \delta/\alpha = \alpha^2/576$, we have}
	$\gamma \leq \beta^2/36$, satisfying the assumptions of \Cref{lem:overlap_prune_clique}. Thus the returned $W'$ is a clique in $R_{\beta/3}(G')$ and satisfies 
	\[ 
	|W| - |W'| \leq (6\gamma/\beta) |V| \leq (48\delta/\alpha)/(\alpha/4) |V| \leq (\alpha/24) |V|. \]
	This means that 
	$|W' \cap S_{\text{g}}| \geq |S_{\text{g}}|/4 - (\alpha/24)|V| \geq |S_{\text{g}}|/6$. 
	On the other hand, we know that $G$ has $(M, t)$-bounded moments in $k$-sparse directions.
	\Cref{lem:bounded_moments_transfers} implies that $G'$ has moments bounded by $O(2^t M /\alpha^3)$. Applying the lemma once more  implies that $R_{\beta/3}(G')$ has
	moments bounded by $O(4^t M/(\alpha^3 \beta)) = O(4^t M/\alpha^4)$. 
	Since $W'$ is a clique in $R_{\beta/3}(G')$, we have by \Cref{lem:clique_round_lemma} that if $\tilde{\mu}$ is the sample mean of $S_{\text{g}} \cap W'$, then 
	\begin{align}\label{eq:almost_there}
		\iprod{ v, \widehat \mu - \tilde \mu }^t \leq O(4^t~M \alpha^{-6}) \quad \text{for all $k$-sparse unit vectors $v$}\;.
	\end{align}
	Since $|S_{\text{g}} \cap W'| \geq |S_{\text{g}}|/6$ and the $S_{\text{g}}$ has bounded $t$-th moment along
	$k$-sparse directions, we have that $\iprod{v, \tilde \mu - \mu_{\text{good}}}^t \leq O(M)$ {(see \Cref{lem:means_close} below for a proof of this fact)}.
	Combining this with \Cref{eq:almost_there} using triangle inequality completes the proof. 
	\qedhere
\end{proof}

We now state and prove \Cref{lem:means_close}. 
	\begin{lemma}\label{lem:means_close}
	Let $t \in \Z_+$ even. Let $\cU$ be a set of unit vectors in $\R^n$ and $S$ be 
	a set with $t$-th moment bounded by $M$ in the directions of $\cU$, i.e., $\E_{X \sim S}[\iprod{v,X-\E_{X \sim S}[X]}^t] \leq M$
	for all $v \in \cU$. Then for all $T \subset S$ with $|T| \geq \alpha |S|$, if we denote 
	by $\mu_S$ and $\mu_T$ the means of $S$ and $T$ respectively, we have that
		$\iprod{v,\mu_S - \mu_T}^t \leq M/\alpha $, 
	for all $v \in \cU$.
\end{lemma}
\begin{proof}
	Let $\mu_S := \E_{X \sim S}[X]$ and $\mu_T := \E_{X \sim T}[X]$. We have that
	\begin{align*}
		M \geq \E_{X \sim S}[\iprod{v,X-\mu_S}^t] 
		\geq \alpha \E_{X \sim T}[\iprod{v,X-\mu_S}^t] 
		\geq \alpha \iprod{v,\mu_T-\mu_S}^t \;,
	\end{align*}    
	where the last inequality uses Jensen's inequality.
\end{proof}

\subsection{Proof of \Cref{thm:LDL-sparse}}

We restate the main theorem below for convenience.
\LDLsparse*

\begin{algorithm}{}
	\begin{algorithmic}[1]
		\Function{LDSparse-Mean}{$S = \{x_1, \ldots, x_m \}, \alpha, M, t, k$}
		\State Let $C \in \Z_+$ be a large enough constant ($C>5$ suffices).
		\State Form the set $T = \{ x-y \; | \; x,y \in S \}$
		\State $T' \gets \textsc{LDMean-SoS-filter}(T,t, C t, M)$
		\State Let $G=(V,E)$ with $V = S$ and $E = \{ (x,y) : \text{$x-y$ or $y-x$ belongs in $T'$} \}$.
		\State $\widehat{\mu} \gets \textsc{Rounding}(S,G)$.
		\State Let $h_k: \R^n \to \R^n$ denote the function where $h_k(x)$ is defined to truncate $x$ to its $k$ largest coordinates in magnitude and zero out the rest.
		\State \Return $h_k(\widehat{\mu})$.
		\EndFunction
	\end{algorithmic}
	\caption{Algorithm for list-decodable sparse mean estimation.}
	\label{alg:ld_sparse_mean}
\end{algorithm}

We start with a brief sketch and provide the full proof below. Let $S$ be the $(1-\alpha)$-corrupted set of samples, and
$S_{\text{g}}$ be the inliers. Let $T := \{ x-y \mid x,y, \in S\}$ and $T_{\text{g}}$ be its part due to inliers, i.e., $T_{\text{g}} := \{ x-y \mid x,y, \in S_{\text{g}}\}$. To every
subset $T'$ of $T$, we can associate a graph $G_{T'}$ having vertices $S$ and edges between the pairs included in $T'$.
Because of \Cref{lem:pop-to-empirical}, $T_{\text{g}}$ has certifiably bounded moments in $k$-sparse directions. By
\Cref{thm:filter_list_decode_mean}, the filtering step will return a subset $T' \subset T$ that has sizable overlap with
$T_{\text{g}}$ and its graph $G_{T'}$ has bounded moments in $k$-sparse directions. Finally, By
\Cref{thm:filter_list_decode_mean}, the rounding algorithm will return a $\widehat{\mu}$ that is close to $\mu$ in all
$k$-sparse directions. This $\widehat{\mu}$ can be truncated to yield a vector close to $\mu$ in the standard $\ell_2$-norm
(\Cref{fact:sparseTruncation}).  

\begin{proof}
	Let $S$ be the $(1-\alpha)$-corrupted set of samples given as input to the algorithm and $S_{\text{g}}$ be the subset of $S$ corresponding to the inliers.
	Given $S$, construct the set of differences $T:= \{x - y \mid x, y \in S\}$. Also,  denote by $T_{\text{g}}$ the same set of corresponding to the inliers.

	For the inliers, the number of samples is large enough so that with constant probability the conclusion of
	\Cref{lem:pop-to-empirical} holds. We thus condition on this event for the rest of the proof. Its first part 
	states that $S_{\text{g}}$ has $(M', t, d)$-certifiable bounded moments in $k$-sparse directions, where $M'=8M$
	and $d=O(t)$.
	By SoS triangle inequality, $T_{\text{g}}$ (\Cref{fact:sos-triangle})  has $(2^t  M', t, O(d))$ bounded moments in $k$-sparse directions.

	Now, \Cref{thm:filter_list_decode_mean} identifies a subset $T' \subset T$ such that with probability at least $2/3$:
	\begin{enumerate}
		\item \label{item:first} For all   $k$-sparse unit vectors $v$ it holds $\sum_{x \in T'} \iprod{v, x}^t \leq 6\cdot 2^tM'|T|$.
		\item \label{item:second} We have $|T'\cap T_{\text{g}}| \geq |T_{\text{g}}|/2$.
	\end{enumerate}
	Construct the graph $G=(V,E)$ with vertex set $V=S$  and edges $(x, y)$ for every pair of $x,y$ that $x-y$ or $y-x$ is in $T'$.  
	By \Cref{item:first} above, $G$  has $(6 \cdot 2^tM',t)$-bounded moments in $k$-sparse directions. 
	By \Cref{item:second}, at least half of the pairs of points in $S_{\text{g}}$  are connected by an edge in $G$.
	Moreover, $S_{\text{g}}$ has $(M',t,d)$-certifiable bounded moments in $k$-sparse directions for $d \geq t$. 
	These are the conditions of \Cref{thm:rounding_list_decode}, thus an application of this to the graph $G$ yields that for every $k$-sparse unit vector $v$ we have that
	\begin{align*}
		\iprod{ v, \widehat \mu - \mu_{\text{g}} }^t = O(10^t M' \alpha^{-6}).
	\end{align*} 
	Also, by the second part of the conclusion of \Cref{lem:pop-to-empirical}, we have that $\iprod{v, \mu - \mu_{\text{g}}}^{t} \leq M'\alpha^{-6}$ for every
	$k$-sparse unit vector $v$. 
	Using the triangle inequality, we have that $\iprod{v, \mu - \widehat \mu}^t \leq O(20^t M'\alpha^{-6})$ for every $k$-sparse unit vector $v$. Then,
	\Cref{fact:sparseTruncation} provides a way to truncate the vector $\widehat{\mu}$ so that the result, $h_k(\widehat{\mu})$, satisfies 
	$\| h_k(\widehat{\mu}) - \mu \|_2^t = O(M'\alpha^{-6}) = O(M \alpha^{-6})$.
	Raising both sides to the power $1/t$ gives the desired claim.
\end{proof}

\section{Information-Computation Tradeoffs} \label{sec:tradeoffs}

In this section we present evidence of an
information-computation gap for our problem, that is,
we provide evidence that computationally efficient list-decoding algorithms for sparse mean estimation of distributions with bounded $t$-th moments up to
error $O(\alpha^{-c/t})$ might inherently need more samples than what is needed to get the same error by 
computationally inefficient algorithms. More specifically, we give statistical query and low-degree polynomial
testing lower bounds for list-decodable sparse mean estimation, which indicate
that the factor $k^{O(t)}$ appearing in the sample complexity of our algorithm from the previous sections might be necessary for computational efficiency.
This is to be compared with the fact that, for distributions with bounded $t$-th moments, 
it is information-theoretically possible to identify a list of $O(1/\alpha)$ candidate vectors, 
containing at least one that is within euclidean distance $O(\alpha^{-1/t})$ using $O(k \log n)/\alpha^3$ samples.
\footnote{This result is shown for the dense case in \cite{DKS18-list}; the adaptation to the sparse case
follows immediately by taking a union bound over the $\binom{n}{k}$ coordinates before applying the VC concentration inequality.} 

In the statistical query model, algorithms are allowed only to perform queries of the following kind instead of drawing samples.

\begin{restatable}[STAT Oracle]{definition}{STATDEF} \label{def:stat}
	Let $D$ be a distribution on $\R^n$. A statistical query is a bounded function $f : \R^n \to [-1,1]$. 
	For $\tau>0$, the $\mathrm{STAT}(\tau)$ oracle responds to the query $f$ with a value $v$ 
	such that $|v - \E_{X \sim D}[f(X)] | \leq \tau$.
	We call $\tau$ the tolerance of the statistical query.
\end{restatable}

The results of this section follow from a simple modification of previous work of \cite{DKS18-list}. 
We thus do not include self-contained proofs here but mention only the key differences. 
We start by formally defining the problem of list-decodable sparse mean estimation. 
Our lower bound would hold against even a weaker noise model, where the noise is i.i.d.\ from an arbitrary distribution.

\begin{problem}[List-Decodable Sparse Mean Estimation] \label{prob:search_problem}
	Fix $\rho > 0$ and $\alpha \in (0,1/2)$.
	Given access to the mixture distribution $\alpha \cN(\rho v, I_n) + (1-\alpha) B$, for some (unknown)
	$k$-sparse unit vector $v$ in $\R^n$ and some (unknown and arbitrary) distribution $B$, the goal is to
	find a list  of vectors $\cL$ with the guarantee that there exists a $u \in \cL$ such that $\|u-\E_{X \sim D}[X]\|_2 < \rho /4$.
\end{problem}

The lower bounds of this section will in fact be about the more basic hypothesis testing version of the problem.

\begin{problem}[Hypothesis Testing of List-Decodable Sparse Means] \label{prob:hypothesis_testing}
	Fix $\rho > 0$.  We define the following hypothesis testing problem:
	\begin{itemize}
		\item $H_0$: The underlying distribution is $\cN(0,I_n)$.
		\item $H_1$: The underlying distribution is $\alpha \cN(\rho v, I_n) + (1-\alpha) B$, for some unknown $k$-sparse unit vector $v$ in $\R^n$ and some unknown distribution $B$.
	\end{itemize}
\end{problem}

It is known that the two problems are related by the following reduction. The resulting algorithm 
is known to be implementable in both the statistical query and the low-degree polynomials model.

\begin{fact}[\cite{diakonikolas2021statistical}]\label{fact:reduction}
	Fix $\rho>0$ and the dimension $n \in \Z_+$. 
	Denote by $\cA$ an algorithm that, whenever given some access to the distribution $\alpha \cN(\rho v, I_n) + (1-\alpha) B$ with unknown $B,v$, it returns a list $\cL$ of candidate vectors such that there exists $u \in \cL$ with $\|u - \rho v \|_2 \leq \rho/4$. 
	Then, there exists a procedure that calls $\cA$ twice and solves the hypothesis testing
	\Cref{prob:hypothesis_testing}
	with probability at least $1 - |\cL|^2/n$. The running time of this reduction is quadratic in $|\cL|n$. 
\end{fact}
\begin{proof}
	We follow the same proof strategy as \cite[Lemma 5.9]{diakonikolas2021statistical} with a crucial modification: instead of using the random rotation matrix $A$ in \cite[Algorithm 1]{diakonikolas2021statistical}, we use a special kind of rotation matrix that can only shuffle the coordinates and flip the signs (see \Cref{fact:random-shuffling} below).
	This modification is needed because the latter family of rotation matrices preserve the sparsity of vectors. 
	We obtain the desired conclusion by following the same proof as \cite{diakonikolas2021statistical} but replacing \cite[Lemma 5.10]{diakonikolas2021statistical} with  the following claim :
	\begin{fact}
		\label{fact:random-shuffling}
		Let $\sigma_1,\dots,\sigma_n$ be $n$ independent Rademacher random variables. Let $A$ be an $n\times n$ independent permutation matrix generated uniformly at random. 
		Let $A'$ be the matrix generated by multiplying the $i$-th row of $A$ by $\sigma_i$ for each $i \in [n]$, i.e, $A'_{i,j} = \sigma_i A_{i,j}$. For any fixed vectors $u$ and $v$, let $Z:= \langle u, A' v\rangle$. Then $\E[Z] = 0$ and the variance of $Z$ is $\|u\|_2^2 \|v\|_2^2/n$.
	\end{fact}
	\begin{proof}
		Let $Z = \langle u, A' v\rangle$ and observe that $Z = \sum_{i,j} \sigma_i u_iA_{i,j}v_j $. Since $\sigma_i$'s are zero mean, we have that $\E[Z] = 0$. 
		To calculate the variance, we use the following facts: (i) 
		For any $i\in[n]$, we have that $A_{i,j}A_{i,\ell} = 0$ almost surely if $j \neq l$, and (ii) $\E[A_{i,j}^2] = \E[A_{i,j}] = 1/n$.
		Using these, we obtain the following expression for the variance of $Z$.
		\begin{align*}
			\Var(Z)    &=  \E\left[\left(\sum_{i,j} \sigma_i u_i  A_{i,j} v_j\right)^2\right]  = \E\left[\sum_{i,j,k,\ell } \sigma_i \sigma_k u_iv_j u_k v_\ell  A_{i,j} A_{k,\ell }\right]\\ 
			&= \E\left[\sum_{i,j,\ell } \sigma_i^2 u_i^2v_j  v_\ell  A_{i,j} A_{i,\ell }\right]   = \E\left[\sum_{i,j }  u_i^2v_j^2   A_{i,j}^2\right]  = \sum_{i,j }  (u_i^2v_j^2)/n  = \|u\|_2^2 \|v\|_2^2/n.
		\end{align*}
		where the third equation is because $\E[\sigma_i\sigma_j] = 0$ if $i \neq j$ and the next one because of $A_{i,j}A_{i,\ell} = 0$ if $j \neq \ell$.
	\end{proof}
	Since the variance is bounded, we can apply the Chebyshev's inequality to get an upper bound on the failure probability of the reduction.
\end{proof}
We now state the SQ lower bound and sketch its proof. 

\begin{theorem}[Statistical Query Lower Bound]
	\label{thm:sq-lower-bd-formal}
	Let $k,n,t \in \Z_+$ with $k \leq \sqrt{n}$, and $c>0$ be a small enough constant. Let $\cA$ be an SQ algorithm that solves the hypothesis testing \Cref{prob:hypothesis_testing} with $\rho = c(t\alpha)^{-1/t}$. Then, $\cA$ does one of the following:
	\begin{itemize}
		\item it uses at least one query with tolerance $O\left( 2^{t/2}k^{-(t+1)/4} \exp\left( O((t\alpha)^{-2/t}) \right)\right)$ or
		\item it makes $\Omega\left(  n^{\sqrt{k}/16} k^{-(t+1)/2} \right)$ many queries.
	\end{itemize}
\end{theorem}
\begin{proof}
	Let $A$ be the one-dimensional distribution of \cite[Lemma 5.5]{DKS18-list}, which satisfies the following properties: (i) $A = \alpha \cN(\rho, 1) + (1-\alpha) E$ for some distribution $E$, (ii)  $A$ matches the first $t$ moments with $\cN(0,1)$, and $\chi^2(A,\cN(0,1)): = \int_{-\infty}^{+\infty} (A(x) - \phi(x))^2/\phi(x) \d x = \exp(O(t\alpha)^{-2/t})$, where $\phi(x)$ denotes the pdf of $\cN(0,1)$.
	Then, the result follows from \cite[Corollary 6.7]{DKKPP22}.
\end{proof}

Instead of using a reduction to the hypothesis testing problem, one can also obtain the same lower bound directly against the list-decodable mean estimation algorithms (i.e., search version of the problem as opposed to the decision version of the problem) by using the framework of \cite{DKS17-sq,DKS18-list}; see, for example, \Cref{thm:sqlb-informal}. The reduction to the hypothesis testing problem outlined here is provided for two reasons:  (i) it is conceptually insightful, and (ii) it allows us to show lower bounds against low-degree polynomial tests (see the remark below).

\begin{remark}
	By using the equivalence between SQ and low-degree polynomials \cite{BBHLS20}, \Cref{thm:sq-lower-bd-formal} also implies qualitatively similar lower bound holds against low-degree polynomial tests. 
	Specifically, the relevant statement for our case is obtained by using  \cite[Theorem 6.23]{DKKPP22} with $m=t$, with the following interpretation:  Unless the number of samples used  is greater than $k^{(1-c)(t+1)}/(2^{t+1}\chi^2(A,\cN(0,I_n))$, any polynomial of degree roughly up to $k^c\log n$  fails to provide a good test for the hypothesis testing problem of \Cref{prob:hypothesis_testing}. We refer to \cite{BBHLS20} for the formal definitions that quantify the notion of goodness of polynomial tests.
\end{remark}

\section{Acknowledgements}

We thank Sihan Liu for pointing out a missing fact in the previous version of the paper.

\bibliographystyle{alpha}
\bibliography{allrefs}

\newcommand{\etalchar}[1]{$^{#1}$}
\begin{thebibliography}{DKK{\etalchar{+}}21b}

\bibitem[ABH{\etalchar{+}}72]{AndBHHRT72}
D.~F. Andrews, P.~J. Bickel, F.~R. Hampel, P.~J. Huber, W.~H. Rogers, and J.~W.
  Tukey.
\newblock {\em Robust Estimates of Location: {{Survey}} and Advances}.
\newblock {Princeton University Press}, {Princeton, NJ, USA}, 1972.

\bibitem[AK05]{arora2005learning}
S.~Arora and R.~Kannan.
\newblock Learning mixtures of separated nonspherical gaussians.
\newblock {\em The Annals of Applied Probability}, 15(1A):69--92, 2005.

\bibitem[BBH{\etalchar{+}}21]{BBHLS20}
M.~Brennan, G.~Bresler, S.~Hopkins, J.~Li, and T.~Schramm.
\newblock Statistical query algorithms and low degree tests are almost
  equivalent.
\newblock In {\em Conference on Learning Theory}, 2021.

\bibitem[BBV08]{BBV08}
M.~F. Balcan, A.~Blum, and S.~Vempala.
\newblock A discriminative framework for clustering via similarity functions.
\newblock In {\em STOC}, pages 671--680, 2008.

\bibitem[BDLS17]{BDLS17}
S.~Balakrishnan, S.~S. Du, J.~Li, and A.~Singh.
\newblock Computationally efficient robust sparse estimation in high
  dimensions.
\newblock In {\em Proc.\ 30th Annual Conference on Learning Theory}, 2017.

\bibitem[BK21]{BK20-subspace}
A.~Bakshi and P.~Kothari.
\newblock List-decodable subspace recovery: Dimension independent error in
  polynomial time.
\newblock In {\em Proceedings of the 2021 ACM-SIAM Symposium on Discrete
  Algorithms (SODA)}, pages 1279--1297. SIAM, 2021.

\bibitem[Bom98]{bomze1998standard}
I.~M. Bomze.
\newblock On standard quadratic optimization problems.
\newblock {\em Journal of Global Optimization}, 13(4), 1998.

\bibitem[BS16]{BarakSteurerNotes}
B.~Barak and D.~Steurer.
\newblock Proofs, beliefs, and algorithms through the lens of sum-of-squares.
\newblock 1, 2016.

\bibitem[CGR16]{CheGR16}
M.~Chen, C.~Gao, and Z.~Ren.
\newblock A general decision theory for {{Huber}}'s \$\textbackslash
  epsilon\$-contamination model.
\newblock {\em Electronic Journal of Statistics}, 10(2):3752--3774, 2016.

\bibitem[CMY20]{cherapanamjeri2020list}
Y.~Cherapanamjeri, S.~Mohanty, and M.~Yau.
\newblock List decodable mean estimation in nearly linear time.
\newblock In {\em 2020 IEEE 61st Annual Symposium on Foundations of Computer
  Science (FOCS)}, pages 141--148. IEEE, 2020.

\bibitem[CSV17]{CSV17}
M.~Charikar, J.~Steinhardt, and G.~Valiant.
\newblock Learning from untrusted data.
\newblock In {\em Proc.\ 49th Annual ACM Symposium on Theory of Computing},
  pages 47--60, 2017.

\bibitem[Das99]{Dasgupta:99}
S.~Dasgupta.
\newblock {Learning mixtures of Gaussians}.
\newblock In {\em Proceedings of the 40th Annual Symposium on Foundations of
  Computer Science}, pages 634--644, 1999.

\bibitem[DG92]{Donoho92}
D.~L. Donoho and M.~Gasko.
\newblock Breakdown properties of location estimates based on halfspace depth
  and projected outlyingness.
\newblock {\em Ann. Statist.}, 20(4):1803--1827, 12 1992.

\bibitem[DK19]{DK19-survey}
I.~Diakonikolas and D.~M. Kane.
\newblock Recent advances in algorithmic high-dimensional robust statistics.
\newblock {\em arXiv preprint arXiv:1911.05911}, 2019.

\bibitem[DKK{\etalchar{+}}16]{DKKLMS16}
I.~Diakonikolas, G.~Kamath, D.~M. Kane, J.~Li, A.~Moitra, and A.~Stewart.
\newblock Robust estimators in high dimensions without the computational
  intractability.
\newblock In {\em Proc.\ 57th IEEE Symposium on Foundations of Computer Science
  (FOCS)}, pages 655--664, 2016.

\bibitem[DKK{\etalchar{+}}19]{DKKPS19-sparse}
I.~Diakonikolas, S.~Karmalkar, D.~M. Kane, E.~Price, and A.~Stewart.
\newblock Outlier-robust high-dimensional sparse estimation via iterative
  filtering.
\newblock In {\em Advances in Neural Information Processing Systems 33, NeurIPS
  2019}, 2019.

\bibitem[DKK20]{diakonikolas2020list}
I.~Diakonikolas, D.~M. Kane, and D.~Kongsgaard.
\newblock List-decodable mean estimation via iterative multi-filtering.
\newblock {\em Advances in Neural Information Processing Systems},
  33:9312--9323, 2020.

\bibitem[DKK{\etalchar{+}}21a]{diakonikolas2022list}
I.~Diakonikolas, D.~M. Kane, D.~Kongsgaard, J.~Li, and K.~Tian.
\newblock Clustering mixture models in almost-linear time via list-decodable
  mean estimation.
\newblock {\em CoRR}, abs/2106.08537, 2021.
\newblock To appear in STOC'22.

\bibitem[DKK{\etalchar{+}}21b]{diakonikolas2021list}
I.~Diakonikolas, D.~M. Kane, D.~Kongsgaard, J.~Li, and K.~Tian.
\newblock List-decodable mean estimation in nearly-pca time.
\newblock {\em Advances in Neural Information Processing Systems}, 34, 2021.

\bibitem[DKK{\etalchar{+}}22]{DKKPP22}
I.~Diakonikolas, D.~M. Kane, S.~Karmalkar, A.~Pensia, and T.~Pittas.
\newblock Robust sparse mean estimation via sum of squares.
\newblock {\em To appear in Conference on Learning Theory, {COLT}}, 2022.
\newblock arXiv preprint arXiv:2206.03441.

\bibitem[DKP{\etalchar{+}}21]{diakonikolas2021statistical}
I.~Diakonikolas, D.~M. Kane, A.~Pensia, T.~Pittas, and A.~Stewart.
\newblock Statistical query lower bounds for list-decodable linear regression.
\newblock {\em Advances in Neural Information Processing Systems}, 34, 2021.

\bibitem[DKS17]{DKS17-sq}
I.~Diakonikolas, D.~M. Kane, and A.~Stewart.
\newblock Statistical query lower bounds for robust estimation of
  high-dimensional {Gaussians} and {Gaussian} mixtures.
\newblock In {\em Proc.\ 58th IEEE Symposium on Foundations of Computer Science
  (FOCS)}, pages 73--84, 2017.

\bibitem[DKS18]{DKS18-list}
I.~Diakonikolas, D.~M. Kane, and A.~Stewart.
\newblock List-decodable robust mean estimation and learning mixtures of
  spherical gaussians.
\newblock In {\em Proceedings of the 50th Annual {ACM} {SIGACT} Symposium on
  Theory of Computing, {STOC} 2018}, pages 1047--1060, 2018.
\newblock Full version available at https://arxiv.org/abs/1711.07211.

\bibitem[DL88]{DonLiu88a}
D.~L. Donoho and R.~C. Liu.
\newblock The "{{Automatic}}" {{Robustness}} of {{Minimum Distance
  Functionals}}.
\newblock {\em The Annals of Statistics}, 16(2):552--586, 1988.

\bibitem[DS07]{dasgupta2007probabilistic}
S.~Dasgupta and L.~J. Schulman.
\newblock A probabilistic analysis of em for mixtures of separated, spherical
  gaussians.
\newblock {\em Journal of Machine Learning Research}, 8:203--226, 2007.

\bibitem[EK12]{eldar2012compressed}
Y.~C. Eldar and G.~Kutyniok.
\newblock {\em Compressed sensing: theory and applications}.
\newblock Cambridge university press, 2012.

\bibitem[FGR{\etalchar{+}}13]{FGR+13}
V.~Feldman, E.~Grigorescu, L.~Reyzin, S.~Vempala, and Y.~Xiao.
\newblock Statistical algorithms and a lower bound for detecting planted
  cliques.
\newblock In {\em Proceedings of STOC'13}, pages 655--664, 2013.
\newblock Full version in Journal of the ACM, 2017.

\bibitem[GLS81]{GLS81}
M.~Gr{\"o}tschel, L.~Lov{\'a}sz, and A.~Schrijver.
\newblock The ellipsoid method and its consequences in combinatorial
  optimization.
\newblock {\em Combinatorica}, 1(2):169--197, 1981.

\bibitem[Hop18]{hopkins2018clustering}
S.~B. Hopkins.
\newblock Clustering and sum of squares proofs: Six blog posts on unsupervised
  learning.
\newblock 2018.

\bibitem[HR09]{Huber09}
P.~J. Huber and E.~M. Ronchetti.
\newblock {\em Robust statistics}.
\newblock Wiley New York, 2009.

\bibitem[HRRS86]{HampelEtalBook86}
F.~R. Hampel, E.~M. Ronchetti, P.~J. Rousseeuw, and W.~A. Stahel.
\newblock {\em Robust statistics. The approach based on influence functions}.
\newblock Wiley New York, 1986.

\bibitem[HTW15]{Hastie15}
T.~Hastie, R.~Tibshirani, and M.~Wainwright.
\newblock {\em Statistical Learning with Sparsity: The Lasso and
  Generalizations}.
\newblock Chapman \& Hall/CRC, 2015.

\bibitem[Kea98]{Kearns:98}
M.~Kearns.
\newblock Efficient noise-tolerant learning from statistical queries.
\newblock {\em Journal of the ACM}, 45(6):983--1006, 1998.

\bibitem[KK10]{kumar2010clustering}
A.~Kumar and R.~Kannan.
\newblock Clustering with spectral norm and the k-means algorithm.
\newblock In {\em 2010 IEEE 51st Annual Symposium on Foundations of Computer
  Science}, pages 299--308. IEEE, 2010.

\bibitem[KKK19]{karmalkar2019list}
S.~Karmalkar, A.~Klivans, and P.~Kothari.
\newblock List-decodable linear regression.
\newblock In {\em Advances in Neural Information Processing Systems}, pages
  7423--7432, 2019.

\bibitem[KKM18]{KlivansKM18}
A.~Klivans, P.~Kothari, and R.~Meka.
\newblock Efficient algorithms for outlier-robust regression.
\newblock In {\em Proc.\ 31st Annual Conference on Learning Theory (COLT)},
  pages 1420--1430, 2018.

\bibitem[KS17a]{KStein17}
P.~K. Kothari and J.~Steinhardt.
\newblock Better agnostic clustering via relaxed tensor norms.
\newblock {\em CoRR}, abs/1711.07465, 2017.

\bibitem[KS17b]{KS17}
P.~K. Kothari and D.~Steurer.
\newblock Outlier-robust moment-estimation via sum-of-squares.
\newblock {\em CoRR}, abs/1711.11581, 2017.

\bibitem[Las01]{lasserre2001new}
J.~B. Lasserre.
\newblock New positive semidefinite relaxations for nonconvex quadratic
  programs.
\newblock In {\em Advances in Convex Analysis and Global Optimization}, pages
  319--331. Springer, 2001.

\bibitem[Li17]{Li17-sparse}
J.~Li.
\newblock Robust sparse estimation tasks in high dimensions.
\newblock {\em CoRR}, abs/1702.05860, 2017.

\bibitem[LRV16]{LaiRV16}
K.~A. Lai, A.~B. Rao, and S.~Vempala.
\newblock Agnostic estimation of mean and covariance.
\newblock In {\em Proc.\ 57th IEEE Symposium on Foundations of Computer Science
  (FOCS)}, pages 665--674, 2016.

\bibitem[MV18]{MeisterV18}
M.~Meister and G.~Valiant.
\newblock A data prism: Semi-verified learning in the small-alpha regime.
\newblock In {\em Conference On Learning Theory, {COLT} 2018}, pages
  1530--1546, 2018.

\bibitem[Nes00]{nesterov2000squared}
Y.~Nesterov.
\newblock Squared functional systems and optimization problems.
\newblock In {\em High performance optimization}, pages 405--440. Springer,
  2000.

\bibitem[RV17]{regev2017learning}
O.~Regev and A.~Vijayaraghavan.
\newblock On learning mixtures of well-separated gaussians.
\newblock In {\em 2017 IEEE 58th Annual Symposium on Foundations of Computer
  Science (FOCS)}, pages 85--96. IEEE, 2017.

\bibitem[RY20a]{raghavendra2020list}
P.~Raghavendra and M.~Yau.
\newblock List decodable learning via sum of squares.
\newblock In {\em Proceedings of the Fourteenth Annual ACM-SIAM Symposium on
  Discrete Algorithms}, pages 161--180. SIAM, 2020.

\bibitem[RY20b]{RY20-subspace}
P.~Raghavendra and M.~Yau.
\newblock List decodable subspace recovery.
\newblock In {\em Conference on Learning Theory, {COLT} 2020}, volume 125 of
  {\em Proceedings of Machine Learning Research}, pages 3206--3226. {PMLR},
  2020.

\bibitem[SCV18]{SteinhardtCV18}
J.~Steinhardt, M.~Charikar, and G.~Valiant.
\newblock Resilience: {A} criterion for learning in the presence of arbitrary
  outliers.
\newblock In {\em Proc.\ 9th Innovations in Theoretical Computer Science
  Conference (ITCS)}, pages 45:1--45:21, 2018.

\bibitem[Sho87]{Sho87}
N.Z. Shor.
\newblock Quadratic optimization problems.
\newblock {\em Soviet Journal of Computer and Systems Sciences}, 1987.

\bibitem[SKL17]{SteinhardtKL17}
J.~Steinhardt, P.~Wei Koh, and P.~S. Liang.
\newblock Certified defenses for data poisoning attacks.
\newblock In {\em Advances in Neural Information Processing Systems 30}, pages
  3520--3532, 2017.

\bibitem[ST21]{steurer2021sos}
D.~Steurer and S.~Tiegel.
\newblock Sos degree reduction with applications to clustering and robust
  moment estimation.
\newblock In {\em Proceedings of the 2021 ACM-SIAM Symposium on Discrete
  Algorithms (SODA)}, pages 374--393. SIAM, 2021.

\bibitem[SVC16]{svc16}
J.~Steinhardt, G.~Valiant, and M.~Charikar.
\newblock Avoiding imposters and delinquents: Adversarial crowdsourcing and
  peer prediction.
\newblock {\em Advances in Neural Information Processing Systems}, 29, 2016.

\bibitem[Tuk60]{Tukey60}
J.~W. Tukey.
\newblock A survey of sampling from contaminated distributions.
\newblock {\em Contributions to probability and statistics}, 2:448--485, 1960.

\bibitem[{van}16]{vandeGeer16}
S.~{van de Geer}.
\newblock {\em Estimation and {{Testing Under Sparsity}}}.
\newblock Number 2159 in \'Ecole d'{{\'Et\'e}} de {{Probabilit\'es}} de
  {{Saint-Flour}}. {Springer}, 1st ed. 2016 edition, 2016.

\bibitem[VW04]{vempala2004spectral}
S.~Vempala and G.~Wang.
\newblock A spectral algorithm for learning mixture models.
\newblock {\em Journal of Computer and System Sciences}, 68(4):841--860, 2004.

\bibitem[Yat85]{Yatracos85}
Y.~G. Yatracos.
\newblock {Rates of convergence of minimum distance estimators and Kolmogorov's
  entropy}.
\newblock {\em Annals of Statistics}, 13:768--774, 1985.

\end{thebibliography}

\newpage

\appendix

\section{Omitted Background}
\label{app:addDetails}

\subsection{Martingales}\label{sec:martingales}

\begin{definition}[Submartingale]
	A submartingale is an integer-time stochastic process $\{X_i \; | \; i \in \Z_+ \}$ that satisfies the following:
	\begin{enumerate}
		\item $\E[|X_i|] < \infty$.
		\item $\E[X_i | X_{i-1},X_{i-2}, \ldots, X_1] \geq X_{i-1}$.
	\end{enumerate}
\end{definition}

\begin{restatable}[Optimal Stopping Theorem]{fact}{OPTIMALSTOPPING} \label{fact:optimalstoping}
	Let $X_1,X_2,\ldots$ be a sub-martingale and $T$ be a finite stopping time. Then, $\E[X_T] \geq \E[X_1]$.
\end{restatable}

\begin{restatable}{proposition}{USEFULFACT}\label{fact:markov_martingales}
	Let $X_1,X_2,\ldots$ be a sub-martingale for which $0 \leq X_i \leq 1$ almost surely and fix an integer $n<\infty$. Then, for any $t \in (0,1)$, we have that
	\begin{align*}
		\pr\left[ \min_{1 \leq i \leq n}X_i \geq t \right] \geq \frac{\E[X_1] - t}{1-t} \;.
	\end{align*}
\end{restatable}
\begin{proof}
	Let the random variable $T$ defined as the minimum between $n$ and $\arg\min_i \{X_i < t\}$. Then $T$ is a stopping time, and it is finite. By the optimal stopping theorem, $\E[X_T] \geq \E[X_1]$. Also, for every random variable $Y \in [0,1]$ and $t\in (0,1)$, Markov's inequality implies that $\pr[ Y \geq t] \geq (\E[Y] - t)/(1-t)$. Using the two, we have that
	\begin{align*}
		\pr\left[ \min_{1 \leq i \leq n}X_i < t \right] \leq \pr\left[ \E[X_T] < t \right] < 1- \frac{\E[X_1] - t}{1-t} \;.
	\end{align*}
\end{proof}

The following lemma has its proof in~\cite{DKKPP22}. Using this, it is possible to show that $O(t^2 \log(n))$ moments being bounded in the standard basis directions is sufficient to show the concentration of the $t$-th tensors in $\ell_\infty$ norm. 

\begin{restatable}{lemma}{BadicLinfConscFull}\label{lem:basic_linf_consc-full}
	Let $D$ be a distribution over $\R^n$ with mean $\mu$. Suppose that for all $s \in [1,\infty)$, $D$ has its $s^{th}$ moment bounded by $(f(s))^s$ for some non-decreasing function $f:[1,\infty) \to \R_+$, in the direction $e_j$, i.e., suppose that for all $j \in [n]$ and $X \sim D$:
	\begin{align*}
		\|\iprod{e_j, X - \mu}\|_{L_s} \leq f(s).
	\end{align*}
	Let $X_1, \dots, X_m$ be $m$ i.i.d.\ samples from $D$ and define $\overline{\mu}:= \sum_{i=1}^m X_i$. The following are true:
	\begin{enumerate}
		\item If $m \geq C \max\left(\frac{1}{\delta^2}, 1\right) \cdot \left( t \log (n/\gamma)\right)\left(2f(t^2\log(n/\gamma))\right)^{2t} \max\left( 1, \frac{1}{f(t)^{2t}}  \right)$, then with probability $1 -\gamma$, we have that 
		\begin{align*}
			\left \| \E_{i \sim [m]}[(X_i - \overline{\mu})^{\otimes t}] -  \E_{X \sim D}[(X - \mu)^{\otimes t}] \right \|_{\infty} \leq \delta \;.
		\end{align*}

		\item If
		$m  > C (k/ \delta^2)  \log(n/\gamma) ( f(\log(n/\gamma))  )^2 $,
		then with probability $1 - \gamma$, it holds 
		\begin{align*}
			\left\| \E_{X \sim S}[X]- \mu \right\|_{2,k} \leq \delta.
		\end{align*}
	\end{enumerate}
\end{restatable}

\end{document}